\documentclass[lettersize,journal]{IEEEtran}
\ifCLASSINFOpdf
\else
\fi

\usepackage{tabularx}
\usepackage{graphics}
\usepackage{bbm}
\newtheorem{theorem}{{Theorem}}
\newtheorem{lemma}{{Lemma}}

\usepackage{stfloats}

\usepackage{CJK}
\usepackage{algorithm}
\usepackage{algorithmic}
\usepackage{amsmath}
\usepackage{amssymb}
\usepackage{psfrag}
\usepackage{graphicx}
\usepackage{epstopdf}
\usepackage{multirow}
\usepackage{stfloats}
\usepackage{cite}
\usepackage{subcaption}
\usepackage{color}
\usepackage[normalem]{ulem}
\usepackage{arydshln}
\usepackage{enumerate}
\usepackage{bbm}
\usepackage{verbatim}
\usepackage{gensymb}

\newtheorem{proposition}{{Proposition}}[section]
\newtheorem{corollary}{Corollary}[section]
\newtheorem{remark}{{Remark}}

\title{On the Capacity Region of Additive-Multiplicative MAC with Heterogeneous Input Constraints}

\author{Qianqian Zhang and Ying-Chang Liang, \emph{Fellow, IEEE} \\

\thanks{
Q.~Zhang and Y.-C. Liang are with University of Electronic Science and Technology of China (UESTC), Chengdu 611731, China (e-mails: qqzhang@uestc.edu.cn and liangyc@ieee.org).

}}

\begin{document}

\maketitle

\vspace{-3em}
\begin{abstract}

This paper characterizes the capacity region of a two-user additive-multiplicative multiple access channel (AM-MAC) under heterogeneous input constraints. This model captures the fundamental limits of symbiotic radio, where an active primary transmitter (PT) conveys information via active transmission subject to an average power constraint, while a passive backscatter device (BD) modulates signals through backscattering under a peak amplitude constraint.
Our main results are threefold. Firstly, we prove that the sum-rate capacity equals the PT's point-to-point capacity, achieved when the PT employs Gaussian signaling and the BD acts as a pure reflector to assist the PT's transmission.
Secondly, to achieve the BD's maximum achievable rate, the PT must adopt a constant-envelope signaling strategy, while the optimal BD distribution exhibits a concentric-circle structure with a uniform phase.
Thirdly, for the remaining boundary points, we establish that the optimal PT signal consists of a continuous uniform phase and a discrete amplitude, whereas the optimal BD distribution is fully discrete.
Finally, numerical results are provided to characterized the capacity region by solving a specialized nonlinear optimization problem. To demonstrate the practical implications, we also characterize an baseline rate pair and evaluate the overall performance of the AM-MAC.

\end{abstract}


\begin{IEEEkeywords}
Multiple access channel, average power constraint, peak amplitude constraint, additive-multiplicative coupling, symbiotic radio.
 \end{IEEEkeywords}

\section{Introduction}

The multiple access channel (MAC) is a fundamental information-theoretic model in wireless networks that enables simultaneous communication from multiple users to a common receiver. In its canonical form, the Gaussian MAC models independent sources transmitting to a common receiver over a shared wireless channel via linear superposition. Under the assumption of additive white Gaussian noise (AWGN), the capacity region defined by the set of all simultaneously achievable rate pairs is elegantly characterized by a polymatroid structure, which manifests as a pentagonal region in the two-user case~\cite{cover1999elements}. This landmark result has long informed the design of modern multiple access schemes, spanning from traditional orthogonal access to emerging non-orthogonal paradigms.
However, the conventional Gaussian MAC framework, premised on the principle of linear superposition, is increasingly insufficient to capture the non-linear and multiplicative interactions inherent in the sixth generation (6G) mobile network architectures and the Internet of Everything (IoE) deployments. 

Among the emerging 6G architectures, symbiotic radio (SR) has emerged as a transformative paradigm~\cite{nawaz2020non,chen2020vision,xiaohutowards}. In SR, a passive backscatter device (BD) modulates information bits by periodically backscattering the signal from a primary transmitter (PT)~\cite{liang2020symbiotic}. This creates a unique additive-multiplicative coupling at the receiver, where the received signal consists of a direct-link additive component and a backscatter-link multiplicative component that the PT's signal acts as the carrier for the BD's data. 
This unique additive-multiplicative coupling creates a mutually beneficial relationship, where the BD gains a free carrier for transmission, while the PT may benefit from additional multipath gain. 

From an information-theoretic perspective, SR forms a new channel model called additive-multiplicative MAC (AM-MAC). Beyond SR, this additive-multiplicative structure is intrinsically shared by other passive-reflection technologies, most notably reconfigurable intelligent surfaces (RIS) when employed for information transmission~\cite{zhang2021reconfigurable,9133134}.
On the other hand, the characterization of SR and RIS-aided systems should consider the heterogeneity of their input constraints. Specifically, the PT is typically a active node subject to an average power constraint, while the passive BD or RIS, due to its hardware limitations and reflection mechanism, is governed by a peak amplitude constraint. Due to these asymmetric constraints, the additive-multiplicative coupling causes the capacity region to diverge from its classic polytopic structure.

This paper provides a comprehensive characterization of the AM-MAC capacity region under heterogeneous input constraints. As shown in Fig.~\ref{fig:system},  the mathematical input-output relation is given by 
\begin{align}
 \label{eq:system model}
{Y} &= a X_1 +  X_1X_2 + {Z},
\end{align}
where $X_1, X_2 \in \mathbb{C}$ represent signals from PT and BD respectively, and $Z \sim \mathcal{CN}(0, \sigma^2)$ represents AWGN. The coefficient $a \in \mathbb{R}\setminus\{0\}$ characterizes the relative channel strength of the direct link compared to the cascaded backscatter link by assuming that any phase offset between the direct and backscatter signals is perfectly compensated. 
Considering the hardware characteristics, the PT is subject to an average power constraint $\mathbb{E}[|X_1|^2]\leq P$, while the passive BD operates under a peak amplitude constraint $|X_2|\leq1$, where $P$ is the maximum transmit power at PT. This heterogeneous constraint structure, coupled with the input-output relation in~\eqref{eq:system model},  distinguishes our model from the canonical additive MAC, i.e., $Y = X_1 + X_2 + Z$~\cite{cover1999elements}.

The investigation of additive-multiplicative coupling has evolved through several distinct stages.
In the degenerate case where $a = 0$~in \eqref{eq:system model}, the model reduces to a purely multiplicative MAC, which arises naturally in passive communications where the direct link is blocked. As for multiplicative MACs, early information-theoretic foundations for noiseless binary case are established in~\cite{cover1999elements}, while the AWGN case is later characterized in~\cite{pillai2011capacity} under homogeneous average power constraints.
Then, the capacity region of the purely multiplicative case with heterogeneous input constraints is investigated in~\cite{zhang2025on}, which reveals that the optimal PT amplitude follows a continuous Rayleigh distribution for sum-rate maximization, while other boundary points require discrete distributions. Beyond the purely multiplicative model, some efforts have been made to address the more general AM-MAC, i.e., $a \neq 0$. However, these studies often rely on simplifying assumptions. For instance, the analyses in~\cite{9518055} and~\cite{liu2018backscatter} are confined to pre-defined, finite-size constellations, instead of characterizing the fundamental capacity region under arbitrary input distributions.
In~\cite{9530367}, the information of the passive device is implicitly embedded in reflection patterns rather than being formulated as an explicit AM-MAC model. Furthermore, while the vector model in~\eqref{eq:system model} is consistent with the framework in~\cite{cheng2021degree}, that work focuses primarily on characterizing the degrees-of-freedom (DoF) in high signal-to-noise ratio (SNR) regimes, leaving the exact characterization of the capacity region as an open and critical problem. 


\begin{figure*}[t]
\centering
\includegraphics[width=1.25\columnwidth] {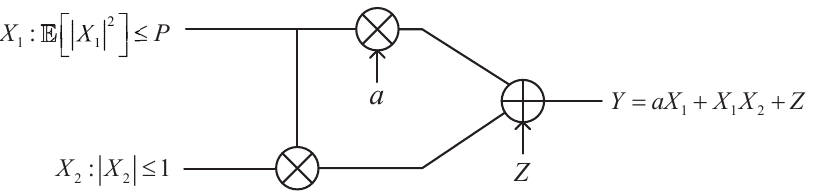}
\caption{A two-user AM-MAC model with heterogeneous input constraints.}
\label{fig:system}
\vspace{-1em}
\end{figure*}

The investigation into capacity-achieving distributions under peak amplitude constraints originates from the foundational study by Smith~\cite{smith1971information}. For a point-to-point AWGN channel subject to both peak amplitude and average power constraints, the optimal distribution is discrete consisting of a finite number of mass points~\cite{smith1971information}. 
This analytical framework is subsequently extended to the complex AWGN channel by Shamai and Bar-David~\cite{shamai1995capacity}. The results establish that the capacity-achieving distribution is characterized by a uniform independent phase and a discrete amplitude, which geometrically manifests as a finite set of concentric circles in the complex plane. More recently, the $n$-dimensional generalization of this problem is addressed in~\cite{dytso2019capacity}. The results confirm that the optimal distribution resides on a finite number of concentric $n$-spheres. 
Further, \cite{9518071} characterizes the AWGN MIMO capacity under a Euclidean-norm peak constraint and derives a rigorous lower bound on the required number of mass points.
Both upper and lower capacity bounds for multiple-antenna systems under amplitude-limited input constraints is characterized in~\cite{7470235}.
Subsequently, upper bounds on capacity for both scalar and vector AWGN channels with peak amplitude constraint are derived in~\cite{thangaraj2017capacity}.
Beyond standard AWGN scenarios, a generalized framework is established in~\cite{8070383} to prove the optimality of discrete distributions for fading channels and signal-dependent noise models under amplitude constraints. 
In~\cite{1435651}, a rigorous investigation into conditionally Gaussian channels demonstrates the necessary and sufficient conditions for discrete measures under joint bounded-input and average-cost constraints.
The study in~\cite{10251160} focuses on a practical hardware modeling and characterizes the capacity of point-to-point backscatter communications by incorporating realistic reflection coefficient constraints derived from load circuit characteristics.
For the additive Gaussian MAC, foundational studies in~\cite{mamandipoor2014capacity,ozel2012capacity} have shown that the entire boundary of the capacity region is achieved by discrete input distributions with a finite support. This fundamental discreteness contrasts with the continuous Gaussian codebooks commonly employed in standard MAC analysis.

This paper provides a rigorous and complete characterization of the capacity region for a two-user AM-MAC subject to heterogeneous input constraints. 
Our contributions are categorized into the following three fundamental aspects. Firstly, we investigate the sum-rate capacity of the AM-MAC. Our analysis reveals a unique phenomenon that the maximum sum-rate of the AM-MAC coincides with the point-to-point capacity of the PT link alone, provided the BD operates in a purely assistive mode.
Secondly, we characterize the boundary point corresponding to the maximum achievable rate of the BD, which requires that the PT adopts a constant-envelope signaling strategy and the optimal BD distribution possesses a concentric-circle structure with a uniform phase distribution, aligning with the classical results for peak-limited quadrature channels. 
Thirdly, for the remaining boundary points that define the trade-off between the PT and BD rates, we provide a deep structural analysis of the optimal input distributions. We establish that the optimal PT signal consists of a continuous uniform phase over $[-\pi, \pi)$ but, crucially, a discrete amplitude distribution. This discretization of the PT's amplitude is a direct consequence of the multiplicative coupling from the peak-limited BD. We rigorously prove that the optimal BD distribution is fully discrete in the complex plane, with its support comprising a finite number of isolated mass points.
Finally, numerical results are provided to characterize the capacity region of the AM-MAC under heterogeneous input constraints. 

The organization of this paper is as follows. Section II demonstrates the main results of the AM-MAC capacity region under heterogeneous constraints, encompassing the converse proof and the construction of optimal input distributions. Section III demonstrates the numerical visualization of the characterized region. Finally, conclusions are provided in Section IV.

The notations used in this paper are summarized as follows. The sets of real numbers and complex numbers are denoted by $\mathbb{R}$ and $\mathbb{C}$, respectively. Other sets are represented by calligraphic letters, e.g., $\mathcal{X}$.
For a complex number $X$, $|X|$, $\angle X$, and $\Re(X)$ denote its magnitude, phase, and real part, respectively. 
The conjugate transpose of a vector $\mathbf{x}$ is denoted by $\mathbf{x}^H$.
The statistical expectation is denoted by $\mathbb{E}[\cdot]$, and the probability density function (PDF) of a random variable $X$ is represented by $f_X(\cdot)$.
We denote the differential entropy of a continuous random variable $X$ by $h(X)$, and the mutual information between random variables $X$ and $Y$ by $I(X; Y)$.
For a subset $\mathcal{A} \subseteq \mathbb{R}^n$, $\mathrm{cl}(\mathrm{conv}(\mathcal{A}))$ denotes the closure of its convex hull. The symbol $\triangleq$ is used for definitions.
For distributions, $\mathcal{CN}(\mu, \sigma^2)$ represents the complex Gaussian distribution with mean $\mu$ and variance $\sigma^2$, and $\mathcal U(a, b)$ denotes the uniform distribution over the interval $(a, b)$.

\section{The Main Results}

Let $R_1$ and $R_2$ denote the achievable transmission rates of the PT and BD, respectively. 
For a vanishingly small average error probability, the capacity region of AM-MAC, denoted by $\mathcal C_{\textrm{AM-MAC}}$, is the closure of the convex hull of all rate pairs $(R_1,R_2)$ for which there exist independent input distributions satisfying the constraints $\mathbb{E}[|X_1|^2] \leq P$ and $|X_2| \leq 1$. Formally, we have
\begin{align*}
&\mathcal C_{\textrm{AM-MAC}}=\mathrm{cl}\left(\mathrm{conv}\left(\bigcup_{\mathbb E|X_1|^2\leq P, |X_2|\leq 1}\!\!\!\mathcal R(R_1,R_2;X_1, X_2)\!\right)\right),
\end{align*}
where $\mathcal R(R_1, R_2;X_1, X_2)$ is the set of all tuples $(R_1,R_2)$ satisfying:
\begin{align}
R_1 & \leq I(X_1;Y|X_2), \label{eq:system model formula1}\\
R_2 & \leq I(X_2;Y|X_1),\label{eq:system model formula2}\\
R_1+R_2 & \leq I(X_1,X_2;Y) \label{eq:system model formula3}.
\end{align}

In what follows, we provide a rigorous characterization of the AM-MAC capacity region by presenting both the converse analysis and the achievability under these heterogeneous input constraints.

\subsection{Converse Analysis}\label{sec:converse}

To establish the converse, we delineate the outer bound of the AM-MAC capacity region under heterogeneous input constraints. This involves proving that any achievable rate pair $(R_1, R_2)$ for AM-MAC must satisfy the inequalities in \eqref{eq:system model formula1}, \eqref{eq:system model formula2}, and \eqref{eq:system model formula3} for some independent input distributions of $X_1$ and $X_2$ subject to the respective average power and peak amplitude limits.

Let $\mathcal C_1^{(n)}$ and $\mathcal C_2^{(n)}$ denote the sequences of codebooks for PT and BD, respectively. We assume these codebooks achieve a rate pair
$(R_1, R_2)$ with a vanishingly small average error probability as the codeword length $n\rightarrow\infty$.
By invoking Fano's inequality and following the standard converse derivation for multiple access channels~\cite[pp. 539]{cover1999elements}, the achievable rates are upper-bounded by the following empirical mutual information terms:
\begin{align}
R_1 &\leq \frac{1}{n}\sum_{\ell=1}^{n}I(X_{1,\ell}^{(n)};Y_{\ell}^{(n)}|X_{2,\ell}^{(n)})+\epsilon_n,\label{eq:1}\\
R_2 & \leq \frac{1}{n}\sum_{\ell=1}^{n}I(X_{2,\ell}^{(n)};Y_{\ell}^{(n)}|X_{1,\ell}^{(n)})+\epsilon_n,\label{eq:2}\\
R_1+R_2 & \leq \frac{1}{n}\sum_{\ell=1}^{n}I(X_{1,\ell}^{(n)},X_{2,\ell}^{(n)};Y_{\ell}^{(n)})+\epsilon_n,\label{eq:3}
\end{align}
where $\epsilon_n = 0 $ as $n\rightarrow \infty$. Here, $\{X_{1,\ell}^{(n)}\}_{\ell=1}^n$ and $\{X_{2,\ell}^{(n)}\}_{\ell=1}^n$ denote the symbols transmitted at time index $\ell$ within a block of length $n$. These symbols are statistically independent for each $\ell$, as the PT and BD encode their respective messages without coordination. The PT's codeword $\{X_{1,\ell}^{(n)}\}_{\ell=1}^n$ belongs to the codebook $\mathcal C_1^{(n)}$, satisfying the average power constraint $\mathbb E [|X_{1,\ell}^{(n)}|^2]\leq P$, while the BD's codeword $\{X_{2,\ell}^{(n)}\}_{\ell=1}^n$ in $\mathcal C_2^{(n)}$ subject to the peak amplitude constraint $|X_{2,\ell}^{(n)} |\leq 1$ for all $\ell \in \{1, \dots, n\}$.
The received signal at the $\ell$-th time slot is modeled as
\begin{align}
Y_{\ell}^{(n)} =(a + X_{2,\ell}^{(n)})X_{1,\ell}^{(n)} + Z_{\ell}^{(n)},
\end{align}
where $Z_{\ell}^{(n)}\sim\mathcal{CN}(0,\sigma^2)$ represents the independent and identically distributed (i.i.d.) complex AWGN at the receiver.

To further simplify the bounds in~\eqref{eq:1},~\eqref{eq:2}, and~\eqref{eq:3}, we introduce a time-sharing random variable $Q$, which is uniformly distributed over the index set $\{1, 2, \dots, n\}$ and independent of the source messages. By conditioning on $Q$, the average of the mutual information terms over $n$ time slots can be expressed in a single-letter form. Specifically, \eqref{eq:1}, \eqref{eq:2}, and \eqref{eq:3} are reformulated as:
\begin{align}
R_1 &\leq I(X_{1,Q}; Y_Q \mid X_{2,Q}, Q) + \epsilon_n \\
R_2 &\leq I(X_{2,Q}; Y_Q \mid X_{1,Q}, Q) + \epsilon_n \\
R_1 + R_2 &\leq I(X_{1,Q}, X_{2,Q}; Y_Q \mid Q) + \epsilon_n,
\end{align}
where $X_{1,Q}$, $X_{2,Q}$, and $Y_Q$ are random variables whose joint distribution, conditioned on $Q = \ell$, is identical to the distribution of the $\ell$-th symbols $(X_{1,\ell}^{(n)}, X_{2,\ell}^{(n)}, Y_{\ell}^{(n)})$ for $\ell \in \{1, \dots, n\}$.

The introduction of $Q$ represents a time-sharing strategy over $n$ signaling points. Since the set of all achievable rate pairs is defined as a convex set, any point generated by such a time-sharing procedure conditioned on $Q$ must, by definition, lie within the capacity region. Consequently, the rate pair $(R_1 - \epsilon_n, R_2 - \epsilon_n)$ is contained within the convex hull of the achievable region defined by the mutual information constraints in \eqref{eq:system model formula1}--\eqref{eq:system model formula3}. As the block length $n \to \infty$, the overhead term $\epsilon_n$ vanishes, i.e., $\lim_{n \to \infty} \epsilon_n = 0$. It follows that any achievable rate pair $(R_1, R_2)$ must lie within the closure of $\mathcal{C}_{\textrm{AM-MAC}}$, which establishes the converse.


\subsection{Achievability}

As detailed in Appendix~\ref{proof:opt}, the boundary points of AM-MAC capacity region under heterogeneous input constraints can be characterized by maximizing the weighted sum rate $\mu_1R_1+\mu_2R_2$ over all achievable rate pairs. Here, $(\mu_1,\mu_2)$ are nonnegative weights satisfying $\mu_1+\mu_2 = 1$. The optimization problem $\textbf{P1}$ is formulated as

\begin{align*}
\textbf{P1}:&
\max_{(R_1,R_2)\in\mathcal C_{\textrm{AM-MAC}}} \mu_1 R_1 + \mu_2 R_2  \\
=& \max_{X_1:\mathbb{E}[|X_1|^2]\leq P,X_2:|X_2|\leq1} g(\mu_1, \mu_2; X_1, X_2),
\end{align*}
where the objective function $g(\cdot)$ is defined based on the relative magnitudes of the weights:
\begin{equation*}
g = \left\{
\begin{aligned}
& I(X_1, X_2; Y), && \text{if } \mu_1 = \mu_2, \\
& \mu_1 I(X_1; Y | X_2) + \mu_2 I(X_2; Y), && \text{if } \mu_1 > \mu_2, \\
& \mu_1 I(X_1; Y) + \mu_2 I(X_2; Y | X_1), && \text{if } \mu_1 < \mu_2,
\end{aligned}
\right.
\end{equation*}
subject to the constraints $\mathbb{E}[|X_1|^2] \leq P$ and $|X_2| \leq 1$.


The optimization in $\textbf{P1}$ is generally non-convex due to the multiplicative coupling between $X_1$ and $X_2$ in the mutual information terms. To circumvent the resulting analytical challenges, we leverage the Lagrangian duality and the Identity Theorem to rigorously characterize the properties of the optimal input distributions.
The resulting capacity region is illustrated in Fig.~\ref{fig:region}. In the following, the boundary of $\mathcal{C}_{\textrm{AM-MAC}}$ is characterized by analyzing the optimal solutions to $\textbf{P1}$ across three distinct weighting scenarios, i.e., $\mu_1 = \mu_2$, $\mu_1 > \mu_2$, and $\mu_1 < \mu_2$.

\begin{figure}[t]
\centering
\includegraphics[width=.65\columnwidth] {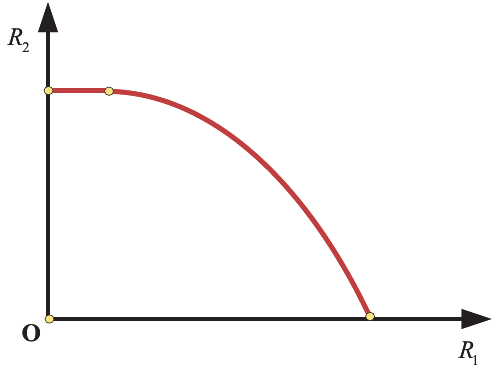}
\caption{Capacity region for model~\eqref{eq:system model} with heterogeneous input constraints.}
\label{fig:region}
\vspace{-1em}
\end{figure}

\subsubsection{Case for $\mu_1=\mu_2$}\label{sec:Rsum}

By defining $X \triangleq (a+X_2)X_1$, the signal model simplifies to $Y = X + Z$. Given the constraints $\mathbb{E}[|X_1|^2] \leq P$ and $|X_2| \leq 1$, the average power of the composite signal $X$ is expressed as
\begin{equation} \label{eq:sum2}
\mathbb{E}[|X|^2] = a^2 \mathbb{E}[|X_1|^2] + \mathbb{E}[|X_1 X_2|^2] + 2a \mathbb{E}[|X_1|^2 \Re(X_2)].
\end{equation}
To maximize the achievable sum rate, it is desirable to maximize the variance of $X$ while ensuring its distribution is as close to Gaussian as possible, since the Gaussian distribution maximizes differential entropy for a fixed power constraint~\cite{cover1999elements}. Consequently, we aim to design input distributions such that $X$ is Gaussian with the maximum possible average power.

Under the peak amplitude constraint $|X_2| \leq 1$, it follows that $\Re(X_2) \leq 1$, which implies $\mathbb{E}[\Re(X_2)] \leq 1$. 
From~\eqref{eq:sum2}, the average power of the composite signal $X$ is maximized when $\mathbb{E}[\Re(X_2)] = 1$ and the PT operates at its full power budget, i.e., $\mathbb{E}[|X_1|^2] = P$. 
Notably, the equality $\mathbb{E}[\Re(X_2)] = 1$ holds if and only if $X_2 = 1$ almost surely, suggesting that any randomness or modulation in $X_2$ would strictly diminish the effective received power.
Consequently, the sum rate is maximized when the BD adopts a purely assistive mode by staying in a fixed reflective state and $X_1 \sim \mathcal{CN}(0, P)$. The maximum sum rate $C_{\text{sum}}$ is thus given by
\begin{equation} \label{eq:sum}
C_{\text{sum}} = \log\left(1 + \frac{P(1+a)^2}{\sigma^2}\right).
\end{equation}

\subsubsection{Case for $\mu_1 >\mu_2$}\label{sec:R1}

We first consider an extreme case where $\mu_1 = 1$ and $\mu_2 = 0$, which corresponds to the maximization of the PT's transmission rate $R_1$.
Let $f_1(X_1)$ and $f_2(X_2)$ denote the PDFs of $X_1$ and $X_2$, respectively. The conditional mutual information $I(X_1; Y | X_2)$ can be expressed as
\begin{align}\label{eq:Jensen_PT}
I(X_1;Y|X_2) &=\int_{\mathbb C}f_2(x_2)I(X_1;Y|X_2 = x_2)dx_2\nonumber\\
&{\leq} \mathop {\max }\limits_{x_2: |x_2| \leq 1} I(X_1;Y|X_2 = x_2),
\end{align}
where the inequality arises because the statistical expectation of a function is upper-bounded by its supremum. 
Equality in~\eqref{eq:Jensen_PT} is attained when $X_2 = 1$ almost surely, which maximizes the effective channel gain.
Under the condition $X_2 = 1$, the composite channel collapses into a point-to-point AWGN channel modeled as $Y = (a+1)X_1 + Z$. 
Given that Gaussian inputs are capacity-achieving for AWGN channels under an average power constraint~\cite{cover1999elements}, the PT reaches its maximum rate $C_1$ when $X_1 \sim \mathcal{CN}(0, P)$, yielding
\begin{equation}
 \label{eq:pri2}
C_1 = \log\left(1+\frac{P(1+a)^2}{\sigma^2}\right).
\end{equation}

For the case where $\mu_1 > \mu_2$, we can reformulate the objective function by substituting $\mu_1 = 1 - \mu_2$, which yields
\begin{align}
& \mu_1 I(X_1; Y | X_2) + \mu_2 I(X_2; Y) \nonumber \\
=& (1 - \mu_2) I(X_1; Y | X_2) + \mu_2 I(X_2; Y) \nonumber \\
=& \mu_2 \left[ I(X_1; Y | X_2) + I(X_2; Y) \right] + (1 - 2\mu_2) I(X_1; Y | X_2) \nonumber \\
=& \mu_2 I(X_1, X_2; Y) + (1 - 2\mu_2) I(X_1; Y | X_2), \label{eq:mu1_greater}
\end{align}
where the last equality follows from the chain rule for mutual information. Since $\mu_1 + \mu_2 = 1$ and $\mu_1 > \mu_2$, it follows that $\mu_2 < \frac{1}{2}$, ensuring that the coefficient $(1 - 2\mu_2)$ is strictly positive.

To maximize the expression in~\eqref{eq:mu1_greater}, we observe that both $I(X_1, X_2; Y)$ and $I(X_1; Y | X_2)$ are simultaneously maximized when $X_2 = 1$ almost surely and $X_1 \sim \mathcal{CN}(0, P)$. In this specific regime, $X_2$ becomes deterministic, causing $I(X_2; Y)$ to vanish and thus $I(X_1, X_2; Y) = I(X_1; Y)$. Consequently, for $\mu_1 > \mu_2$, the optimal signaling strategy collapses to the sum rate optimal case, where the BD acts solely as a passive reflector to facilitate the PT's transmission.

\subsubsection{Case for $\mu_1 < \mu_2$}\label{sec:R2}

\begin{figure}[t]
\centering
\includegraphics[width=.75\columnwidth] {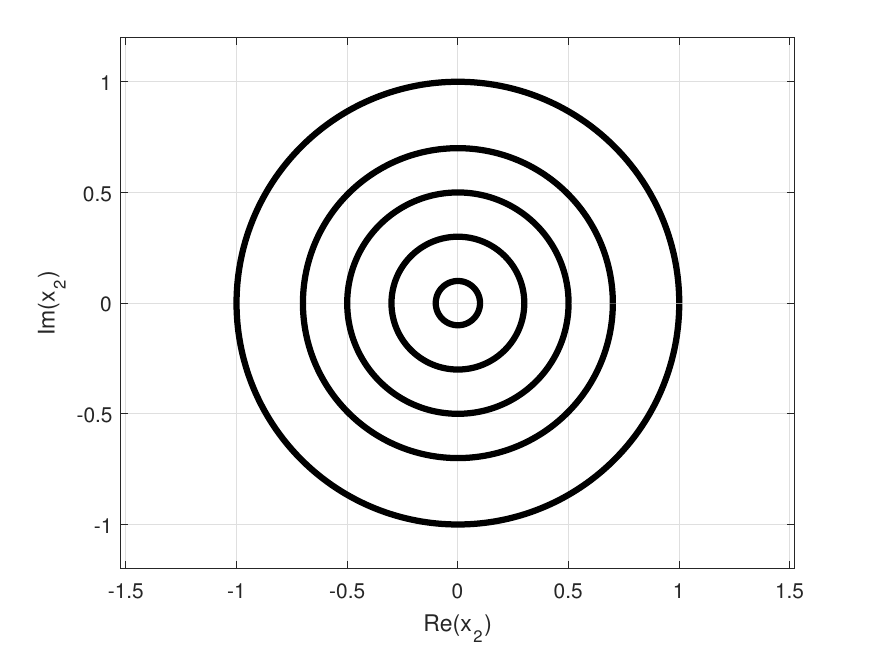}
\caption{An example of the optimal input distribution with peak amplitude constraint.}
\label{fig:distribution}
\vspace{-1em}
\end{figure}

We first consider the special case where $\mu_1 = 0$ and $\mu_2 = 1$, which characterizes the maximum achievable rate of the BD.
\begin{theorem}\label{theorem:BD}
When $\mu_1 = 0$ and $\mu_2 = 1$, the capacity-achieving distributions satisfy the following conditions: (1) The PT transmits constant-envelope signal with $|X_1|^2 = P$ almost surely; and (2) the BD transmits a signal characterized by a discrete amplitude and an independent, uniformly distributed phase.
\end{theorem}
\begin{IEEEproof}
Under the weighting $\mu_1 = 0$ and $\mu_2 = 1$, the objective reduces to maximizing the achievable rate of the BD, $I(X_2; Y | X_1)$.
According to~\cite{shamai1995capacity}, for a given realization $X_1 = x_1$, the distribution of $X_2$ that maximizes the conditional mutual information $I(X_2;Y|X_1 = x_1)$ under the peak amplitude constraint $|X_2|\leq 1$ is discrete in amplitude with a uniform independent phase. Invoking Jensen's inequality yields
\begin{align}
I(X_2;Y|X_1) & = \mathbb E_{X_1}I(X_2;Y|X_1 = x_1)\nonumber\\
&\overset{(a)}{\leq} I(X_2;Y||X_1|^2 = P),
\end{align}
where step $(a)$ is justified by the following three observations.
Firstly, the phase of $X_1$ does not affect the mutual information $I(X_2; Y | X_1 = x_1)$ due to the circular symmetry of the AWGN.
Secondly, as proved in Appendix~\ref{proof:x1concave}, $I(X_2; Y | |X_1|^2 = p)$ is a strictly concave function of the power $p$.
Thirdly, by Jensen's inequality, we have $\mathbb{E}[I(|X_1|^2)] \leq I(\mathbb{E}[|X_1|^2])$, where the equality in the inequality holds when $X_1$ has a constant envelope, i.e., $|X_1|^2 = P$ with probability one.
This concludes the proof.
\end{IEEEproof} 

Next, we investigate the regime where $0 < \mu_1 < \mu_2$. 
In this scenario, the objective function of problem $\textbf{P1}$ can be reformulated as
\begin{align} \label{eq:sympli}
& \mu_1 I(X_1; Y) + \mu_2 I(X_2; Y | X_1) \nonumber \\
=& \mu_1 [h(Y) - h(Y | X_1)] + \mu_2 [h(Y | X_1) - h(Y | X_1, X_2)] \nonumber \\
=& \mu_1 h(Y) + (\mu_2 - \mu_1) h(Y | X_1) - \mu_2 \log(\pi e \sigma^2),
\end{align}
where the second equality utilizes the fact that $h(Y | X_1, X_2) = h(Z) = \log(\pi e \sigma^2)$ due to the additive nature of the noise. 
By analyzing the impact of the PT's phase on the differential entropies $h(Y)$ and $h(Y | X_1)$, we establish the following property.
\begin{lemma}\label{lemma:x1phase}
For the optimization problem $\textbf{P1}$ with $\mu_2>\mu_1$, the optimal input distribution of $X_1$ is characterized by its phase $\theta_1 \triangleq \angle X_1$ that is uniformly distributed over $[-\pi, \pi)$ and statistically independent of its amplitude $r \triangleq |X_1|$.
\end{lemma}
\begin{IEEEproof}
Please refer to Appendix~\ref{proof:x1phase}.
\end{IEEEproof} 

With Lemma~\ref{lemma:x1phase} established, the remaining task is to determine the optimal amplitude distribution of $X_1$ and the optimal distribution of $X_2$. In the following, we leverage the Identity Theorem to prove that these optimal distributions are indeed discrete with finite support.

Let $f_r(r) \in \mathcal{F}_r$ denote PDF of the amplitude of $X_1$. By invoking the phase symmetry established in Lemma~\ref{lemma:x1phase}, problem $\textbf{P1}$ for the regime $\mu_2 > \mu_1$ is simplified to
\begin{align*}
\textbf{P2}: \max_{f_r \in \mathcal{F}r} \quad & \mu_1 h(Y) + (\mu_2 - \mu_1) h(Y | X_1) \\
\text{s.t.} \qquad & \int_{0}^{\infty} r^2 f_r(r) dr \leq P.
\end{align*}

As demonstrated in Appendix~\ref{proof:concave}, the objective function in $\textbf{P2}$ is a strictly concave functional of the input PDF $f_r(r)$. 
Since both the power constraint and the requirement $\int_{0}^{\infty} f_r(r) dr = 1$ are linear functionals of $f_r(r)$, the feasible set is a convex set of measures. Consequently, $\textbf{P2}$ is a convex optimization problem in the space of distributions. 
By leveraging Lagrangian duality theory, the problem can be reformulated into the following form:
\begin{align} \label{eq:ne12}
\min_{\lambda \geq 0} \max_{f_r \in \mathcal{F}_r} \quad & \mathcal{J}(f_r, \lambda),
\end{align}
where $\lambda$ is the Lagrangian multiplier associated with the PT's average power constraint.
The Lagrangian functional $\mathcal{J}(f_R, \lambda)$ is defined over the convex set $\mathcal{F}_r$ and can be expanded as
\begin{align} \label{eq:lagrangian_expanded}
\mathcal{J}(f_r, \lambda) &\triangleq \mu_1 h(Y) + (\mu_2 - \mu_1) h(Y | r) - \lambda \int_{0}^{\infty} r^2 f_r(r) dr \nonumber \\
&= \mu_1 \int_{0}^{\infty} f_r(r) \omega_1(r; f_r) dr + (\mu_2 - \mu_1)\times\nonumber \\
& \quad\int_{0}^{\infty} f_r(r) h(Y |r) dr - \lambda \int_{0}^{\infty} r^2 f_r(r) dr,
\end{align}
where $h(Y | r)$ represents the conditional differential entropy of $Y$ for a fixed PT amplitude $r$. The term $\omega_1(r; f_r)$ captures the contribution of the amplitude distribution $f_r$ to the marginal entropy $h(Y)$, formally defined as
\begin{equation} \label{eq:59}
\omega_1(r; f_r) \triangleq -\int_{\mathbb{C}} K(y, r) \log f_Y(y) dy.
\end{equation}
Here, $K(y, r)$ serves as the effective transition kernel from the PT amplitude $r$ to the received signal $y$, obtained by marginalizing over the uniform phase $\theta_1$ and the BD distribution $f_{X_2}(x_2)$, which is given by
\begin{align} \label{eq:kernel_def}
K(y, r) \triangleq & \frac{1}{2\pi^2\sigma^2} \int_{\mathbb{C}} \int_{-\pi}^{\pi} f_{X_2}(x_2)\times \nonumber\\
&\quad\exp\left( -\frac{|y - (a+x_2)re^{j\theta_1}|^2}{\sigma^2} \right) d\theta_1 dx_2.
\end{align}

To characterize the optimal distribution $f_r^*$, we evaluate the Gateaux derivative (weak derivative) of the functional $\mathcal{J}(f_r, \lambda)$ at $f_r^*$ in the direction of an arbitrary PDF $f_r \in \mathcal{F}_r$, which is defined as
\begin{equation*}
\mathcal{J}_{f_r^{*}}^{'}(f_r) \triangleq \lim_{\epsilon \to 0} \frac{\mathcal{J}((1-\epsilon)f_r^* + \epsilon f_r) - \mathcal{J}(f_r^*)}{\epsilon},
\end{equation*}
where $\epsilon \in [0, 1]$. 
By defining the perturbed distribution $f_r^\epsilon \triangleq (1-\epsilon)f_r^* + \epsilon f_r$, the difference in the functional values can be expanded, which is given by 
\begin{align*}
& J(f_r^{\epsilon})-J(f_r^{*})= \mu_1\int_0^{+\infty} f_r^*(\omega_1(r;f_r^{\epsilon})-\omega_1(r;f_r^*))dr\\
 &+\int_0^{+\infty}\!\!\epsilon(f_r-f_r^{*})(\mu_1\omega_1(r;f_r^{\epsilon})+(\mu_2-\mu_1)h(Y|r)-\lambda r^2)dr
\end{align*}
Accordingly, the weak derivatives is given by
\begin{align*}
J_{f_r^{*}}^{'}(f_r)\!\! =\!\!\! \int_{0}^{\infty}\!\!\!\! (f_r\!-\!f_r^{*})(\mu_1\omega_1(r;f_r^{*})\!+\!(\mu_2\!-\!\mu_1)h(Y|r)\!-\!\lambda r^2)dr.
\end{align*}
Following the approach in~\cite{luenberger1997optimization}, a necessary and sufficient condition for $f_r^*$ to be the maximizer of the concave functional $\mathcal{J}$ is that the Gateaux derivative is non-positive for all $f_r \in \mathcal{F}_r$, i.e., $\mathcal{J}_{f_r^{*}}^{'}(f_r)\leq 0$.
This leads to the following variational inequality
\begin{equation} \label{eq:cond11}
\int_{0}^{\infty} f_r(r) \left[ \omega_1(r; f_r^*) + \left( \frac{\mu_2}{\mu_1} - 1 \right) h(Y |r) - \frac{\lambda}{\mu_1} r^2 \right] dr \leq J_1,
\end{equation}
where $J_1 \triangleq \int_0^\infty f_r^*(r) \left[ \omega_1(r; f_r^*) + \left( \frac{\mu_2}{\mu_1} - 1 \right) h(Y | r) \right] dr - \frac{\lambda}{\mu_1} P$ is a constant independent of $f_r$. Furthermore, the complementary slackness condition ensures that $\lambda (\int_0^\infty r^2 f_r^*(r) dr - P) = 0$. Consequently, the inequality in \eqref{eq:cond11} remains valid regardless of whether the power constraint is active.

By analyzing the variational inequality~\eqref{eq:cond11}, we obtain the following necessary and sufficient conditions for optimality.
\begin{theorem}\label{theorem:feature}
The amplitude distribution $f_r^*$ is optimal if and only if there exists a constant $J_1$ and a Lagrangian multiplier $\lambda \geq 0$ such that the following conditions are satisfied
\begin{align}
& \mathcal L(r; f_r^*) \leq J_1, \quad \forall r \geq 0, \label{eq:condtion11} \\
& \mathcal L(r; f_r^*) = J_1, \quad \forall r \in \mathcal{S}_{f_r^*}, \label{eq:condtion12}
\end{align}
where $\mathcal L(r; f_r^*) \triangleq \omega_1(r;f_r^{*}) + \left( \frac{\mu_2}{\mu_1} - 1 \right) h(Y | r) - \frac{\lambda}{\mu_1} r^2$ is the Lagrangian marginal functional, and $\mathcal S_{f_r^*}$ denotes the support of $f_r^{*}$.
\end{theorem}
\begin{IEEEproof}
The proof of this Theorem follows from the properties of Gateaux derivatives for concave functionals on a convex set of PDFs. Please see Appendix~\ref{proof:condition} for detailed derivations.
\end{IEEEproof}

Theorem~\ref{theorem:feature} provides the necessary and sufficient conditions that any optimal amplitude distribution must satisfy via the variational inequality. However, it remains to be determined whether the optimal radial distribution $f_r^*$ is continuous or discrete. By examining the real-analyticity of this functional and its asymptotic characteristics as $r \to \infty$, we can rigorously characterize the geometry of the optimal signaling, as presented in the following theorem.

\begin{theorem}
For the weighting regime $\mu_2 > \mu_1$, the optimal amplitude distribution $f_r^*$ that maximizes the functional $\mathcal{J}(f_r, \lambda)$ is uniquely supported on a finite discrete set of points within the interval $[0, \infty)$.
\end{theorem}

\begin{IEEEproof}
To establish the structural properties of $f_r^*$, we analyze the analytical behavior of the Lagrangian marginal functional $\mathcal{L}(r; f_r^*)$. As detailed in Appendix~\ref{proof:real}, the constituent terms of $\mathcal{L}(r; f_r^*)$ are shown to be real-analytic on the domain $(0, \infty)$.

Furthermore, by examining its asymptotic characteristics as $r \to \infty$, we demonstrate that $\mathcal{L}(r; f_r^*)$ is dominated by a quadratic term and a logarithmic term, ensuring it cannot be identically zero as $r \to \infty$. By invoking the Identity Theorem for real-analytic functions, we prove that the roots of the equation $\mathcal{L}(r; f_r^*) = J_1$ must be isolated. Detailed derivations are provided in Appendix~\ref{proof:x1}. 
\end{IEEEproof}

\begin{remark}
The discrete nature of the optimal amplitude distribution $f_r^*$ in our AM-MAC model is fundamentally rooted in the channel uncertainty created by the backscatter modulation. In~\cite{abou2001capacity}, it is established that for Rayleigh fading channels without channel state information, the absence of knowledge regarding the fading channel response forces the capacity-achieving input to be discrete.

A rigorous mathematical parallel can be drawn between the non-coherent setting and our $X_1 X_2$ interaction in~\eqref{eq:system model}. Specifically, the BD's input $X_2$ acts as an equivalent unknown channel state from the perspective of decoding $X_1$ at the receiver. 
Just as the stochastic nature of the fading in~\cite{abou2001capacity} and multiplicative phase perturbations in~\cite{1337103,yang2017multiplexing} break the optimality of Gaussian signaling, the multiplicative uncertainty introduced by $X_2$ necessitates the use of discrete probability measures to achieve the capacity boundary of the AM-MAC. This observation aligns our findings with the broader principle that Gaussian inputs are generally suboptimal for channels characterized by multiplicative state uncertainty.   $\hfill\blacksquare$
\end{remark}

Next, we characterize the optimal distribution of $X_2$ for a given $f_{X_1}$. As demonstrated in Appendix~\ref{proof:x2concave}, both $h(Y)$ and $h(Y|X_1)$ are strictly concave functionals of $f_{X_2} \in \mathcal{F}_{X_2}$. 
Since strict concavity is preserved under non-negative linear combinations, the objective functional of problem $\textbf{P1}$ is strictly concave over the space of BD input distributions.

Parallel to the analysis for $f_r(r)$, the distribution $f_{X_2}^*$ is optimal if and only if there exists a constant $T_1$ such that the following KKT conditions are satisfied:
\begin{align}
&\omega_2(X_2;f_{X_2}^{*}) \leq T_1, \quad \forall X_2 \in \mathcal{D}, \label{eq:opt_x2_1}\\
&\omega_2(X_2;f_{X_2}^{*}) = T_1, \quad \forall X_2 \in \mathcal{S}_{f_{X_2}^*}, \label{eq:opt_x2_2}
\end{align}
where $\mathcal{D} \triangleq \{x_2 : |x_2| \leq 1\}$ represents the unit disk constraint for the backscatter signal, and $\mathcal{S}_{f_{X_2}^*}$ denotes the support of $f_{X_2}^*$. The auxiliary functional $\omega_2(x_2; f_{X_2})$ is defined as
\begin{align*}
&\omega_2(x_2; f_{X_2}) \triangleq -\int_{\mathbb{C}} \int_{\mathbb{C}} f_{X_1}(x_1) f_Z(y-(a+x_2)x_1)\times \\
&\quad \left[ \mu_1 \log f_Y(y) + (\mu_2-\mu_1) \log f_{Y|X_1}(y|x_1) \right] dy dx_1,
\end{align*}
where $f_Z(z)$ is the PDF of the noise $Z$.
\begin{theorem}\label{prop:no_2D_support}
For the weighting regime $\mu_2 > \mu_1$, the optimal distribution $f_{X_2}^*$ in the complex plane $\mathbb{C}$ is a discrete random variable uniquely supported on a finite number of mass points within the unit disk $|X_2| \leq 1$.
\end{theorem}
\begin{IEEEproof}
To characterize the support $\mathcal{S}_{f_{X_2}^*}$, we investigate the analytical properties of the auxiliary functional $\omega_2(X_2)$. As established in Appendix~\ref{proof:x2}, $\omega_2(X_2)$ is a real-analytic function of the complex variable $X_2$ over the unit disk $\mathcal{D}$.
The discreteness of the optimal distribution follows from the fact that the gradients of the constituent geometric components of $\omega_2(X_2)$ are linearly independent over $\mathbb{C}$. Combined with the Identity Theorem for real-analytic functions, this property precludes the existence of any continuous level sets or curves where the optimality condition $\omega_2(X_2) = T_1$ could hold. Consequently, the support $\mathcal{S}_{f_{X_2}^*}$ is restricted to a set of isolated points. Details of this proof are in Appendix~\ref{proof:x2}. 
\end{IEEEproof}

\begin{figure}[t]
\centering
\includegraphics[width=.8\columnwidth] {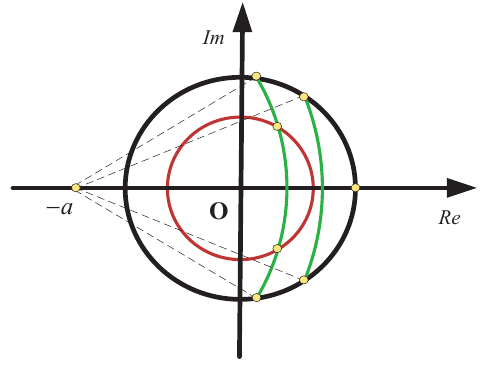}
\caption{Geometric illustration of the level sets of $\omega_2(X_2)$. The intersection of circles centered at the origin and $-a$ characterizes the isolated points of the optimal BD support $\mathcal{S}_{f_{X_2}^*}$.}
\label{fig:design}
\vspace{-2em}
\end{figure}

\begin{remark}
The structural properties of the optimal BD distribution $f_{X_2}^*$ can be intuitively understood from the perspective of geometric symmetries. Specifically, the entropy term $h(Y)$ depends on $X_2$ primarily through the magnitude $|a + X_2|$, which exhibits rotational symmetry along circles centered at $-a$. In contrast, the conditional entropy $h(Y|X_1)$ is governed by $|X_2|$, remaining invariant along circles centered at the origin.
The geometry of the optimal support $\mathcal{S}_{f_{X_2}^*}$ is thus dominated by the interplay between these two distinct families of level sets: circles centered at $-a$ and those centered at the origin. 
Since these two families of circles are non-concentric for $a \neq 0$, this geometric misalignment prevents the optimality condition from being sustained over any continuous area or curve in the complex plane, leading to the discreteness of the optimal input distribution of the BD. $\hfill\blacksquare$
\end{remark}

Fig.~\ref{fig:design} provides a conceptual visualization of the geometric misalignment, which leads to the discreteness of the BD support. 
To clarify that the optimal mass points $\mathcal{S}_{f_{X_2}^*}$ do not merely coincide with the geometric intersections of the two families of circles. Rather, the location of each mass point satisfies that the weighted gradients of the entropy terms $\mu_1 \nabla h(Y)$ and $(\mu_2-\mu_1) \nabla h(Y|X_1)$ cancel each other out perfectly.
As illustrated by the yellow markers in this figure, these points are isolated precisely because the non-concentric nature of the two symmetry centers of $a \neq 0$ prevents any curve that can keep that gradient in balance. 
Thus, the figure illustrates that the support must be discrete, although the exact optimal coordinates are derived through more complex calculations.

%
\section{Numerical Results} \label{sec:results}

In this section, we provide numerical evaluations to characterize the capacity region of the AM-MAC under heterogeneous input constraints. To provide a clear performance benchmark, we first establish a baseline achievable rate pair based on conventional signaling. Subsequently, we detail the numerical optimization framework to determine the Pareto frontier of the capacity region. Finally, we investigate the performance in terms of the achievable rates to key system parameters, including the direct link strength $a$ and the transmit SNR.

\subsection{Baseline Achievable Rate Pair}\label{sec:baseline}

To provide a benchmark for the capacity region, we first establish a baseline achievable rate pair by employing a common signaling scheme. In this baseline, the PT adopts a continuous Gaussian distribution satisfying its average power constraint, i.e., $X_1 \sim \mathcal{CN}(0, {P})$, while the BD employs a circular uniform distribution on the unit circle of the complex plane to satisfy its peak amplitude constraint, i.e., $X_2 = e^{j\theta_2}$ with $\theta_2 \sim \mathcal{U}[0, 2\pi)$. Under this setup, the secondary transmission $X_2$ acts as a pure phase noise from the perspective of the receiver to decode the PT's signal, while $X_1$ provides the carrier for the backscatter link. 

The achievable rate for the PT, denoted by $R_1^{\text{base}}$, is determined by the mutual information $I(X_1; Y)$. By leveraging the circular symmetry of the signal and noise, $R_1^{\text{base}}$ can be expressed as
\begin{align}\label{eq:r1}
&R_1^{\text{base}}=  \mathbb E_{x_1}\left[\int_{0}^{+\infty}\frac{2{\rho}{\kappa}({\rho,x_1})}{\sigma^2} \log\left(\frac{2{\kappa}({\rho,x_1})}{\sigma^2}\right) d{\rho}\right]\nonumber\\
&-\int_{0}^{+\infty}{\rho}\bar{\kappa}(\rho)\log(\bar\kappa(\rho)) d{\rho},
\end{align}
where the conditional density kernel is defined as ${\kappa}(\rho,x_1) = \exp\left(-\frac{{\rho}^2+|x_1|^2}{\sigma^2}\right)I_0\left(\frac{2{\rho}|x_1|}{\sigma^2}\right)$, and $I_0(x) = \frac{1}{2\pi}\int_{-\pi}^{\pi}e^{x\cos\theta}d\theta$, and the marginal density is given by $\bar{\kappa}(\rho) = \int_{-\pi}^{\pi}\frac{1}{\pi(P|a+e^{j\theta_2}|^2+\sigma^2)}\exp\left(-\frac{{\rho}^2}{\sigma^2+P|a+e^{j\theta_2}|^2}\right)d\theta_2$.
By treating $X_1X_2$ as interference, we obtain a lower bound on the achievable rate of the PT, which is given by
\begin{align}\label{eq:lower1}
R_{1,\text{Low}}^{\text{base}} = \log_2\left(1 + \frac{P a^2 }{P+\sigma^2}\right).
\end{align}

The achievable rate of the BD, $R_2^{\text{base}} = I(X_2; Y|X_1)$, corresponds to the capacity of a phase-modulated signal over an AWGN carrier. Given a realization of $X_1$, this is expressed as
\begin{align}\label{eq:r2}
R_2^{\text{base}} = \mathbb E_{x_1}\left[-\int_{0}^{+\infty}\frac{2{\rho}{\kappa}({\rho,x_1})}{\sigma^2} \log\left({\kappa}({\rho,x_1})\right) d{\rho}-\log(e)\right].
\end{align}

To gain further insights into the high and low SNR regimes, we adopt the asymptotic characterization of phase-modulation capacity from~\cite{wyner1966bounds}. The asymptotic behavior of $R_2^{\text{base}}$ can be approximated as follows:
\begin{align}\label{eq:asy1}
R_{2,\text{Asy}}^{\text{base}} \approx \mathbb E_{x_1}\left\{
\begin{array}{cl}
   \frac{1}{2}\log\left(\frac{4\pi |x_1|^2}{e\sigma^2}\right),\;\;&{\frac{|x_1|^2}{\sigma^2}\gg1},    \\
  \frac{|x_1|^2}{\sigma^2},\;\;&{\frac{|x_1|^2}{\sigma^2}\ll1}.
  \end{array}
  \right.
\end{align}\

\begin{figure}[t]
\centering
\includegraphics[width=.88\columnwidth] {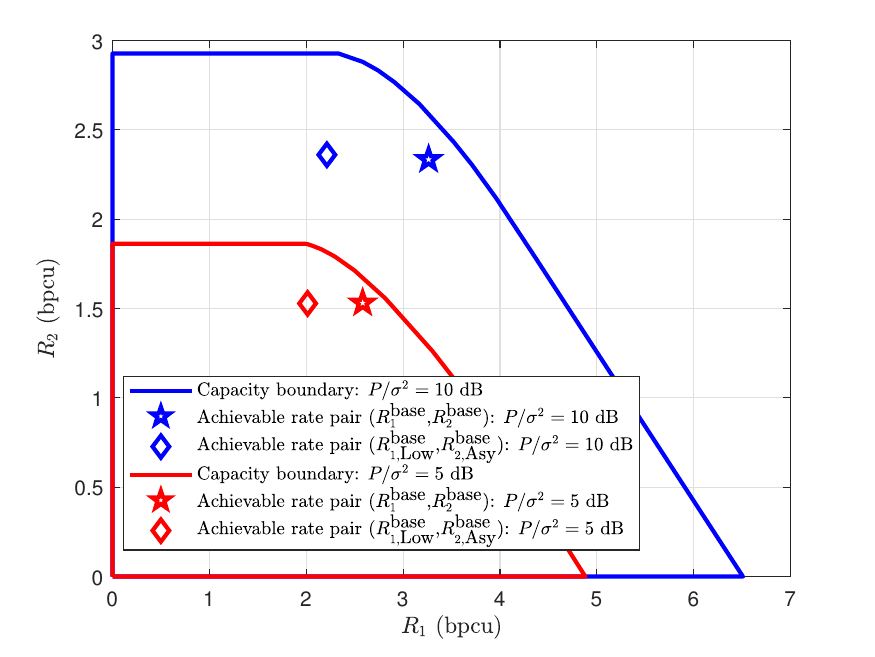}
\caption{Capacity region with $a=2$.}
\label{fig:regionsnr}
\end{figure}

\subsection{Numerical Optimization for the Capacity Region Boundary}

While the baseline rate pair provides an analytical reference, it relies on continuous distributions that, as established in Section~\ref{sec:baseline}, are suboptimal under peak amplitude constraints. 
To determine the exact boundary of the capacity region $\mathcal{C}_{\text{AM-MAC}}$, we employ a numerical optimization framework to identify the optimal discrete mass points and their corresponding probabilities for $r$ and $X_2$.

Specifically, the rate pairs on the capacity boundary can be obtained by solving the following weighted sum-rate maximization problem for various weight vectors $(\mu_1, \mu_2)$:
\begin{align*}
\textbf{P3}:  \;\; \mathop {\max}\limits_{f_r,f_{X_2}}\;\; &\mu_1I(X_1;{Y})+\mu_2I(X_2;{Y}|X_1)\nonumber \\
\textrm{s.t.} \;\;\quad  & \mathbb{E}[r^2]\leq P  ~\textrm{and} ~|X_2|\leq1.
\end{align*}
The optimization variables are the numbers, locations, and the associated probabilities of the mass points of $r$ and $X_2$. Problem $\textbf{P3}$ is a non-convex optimization problem and belongs to nonlinear programming. Nonetheless, it can be solved by using an interior-point algorithm, which introduces barrier functions to deal with inequality constraints and adopts conjugate gradient methods to find a convergence point~\cite{nocedal2014interior}. In the following results, we solve problem $\textbf{P3}$ directly using the fmincon solver in MATLAB, which solves nonlinear programming problems by using interior-point algorithm~\cite{MathWorks}. 

\subsection{Results}

Fig.~\ref{fig:regionsnr} depicts the capacity region of the AM-MAC under heterogeneous constraints for two different SNR levels with $a=2$. As observed in this figure, an increase in the SNR $P/\sigma^2$ leads to a significant expansion of the capacity region in both the $R_1$ and $R_2$ dimensions.
The curved portion of the capacity boundary represents the Pareto optimal trade-off between the two users. The rectangular profile observed at low $R_1$ values indicates that for a certain range, the BD can transmit at its peak rate without significantly degrading the PT's performance.
The capacity boundaries, obtained by solving $\textbf{P3}$ through the fmincon solver, are compared against the baseline signaling pairs in~\eqref{eq:r1} and~\eqref{eq:r2} and their corresponding analytical approximations in~\eqref{eq:lower1} and~\eqref{eq:asy1}.
We can find that the achievable rate pairs $(R_1^{\text{base}}, R_2^{\text{base}})$, represented by star markers, strictly lie within the capacity region. This gap demonstrates the shaping gain achieved by the optimized discrete distributions over conventional continuous distributions.
Several key insights emerge from the comparison between the baseline pairs and their analytical simplifications.
Firstly, the asymptotic rate $R_{2,\text{Asy}}^{\text{base}}$ provides a remarkably accurate approximation of the true baseline rate $R_2^{\text{base}}$ at both SNR levels. 
In contrast, a noticeable gap exists between the true baseline rate $R_1^{\text{base}}$ and its conservative lower bound $R_{1,\text{Low}}^{\text{base}}$. This occurs because the lower bound treats the backscatter signal strictly as additive interference, neglecting the fact that the receiver can partially exploit the structured nature of the multiplicative coupling to improve decoding performance.


\begin{figure}[t]
\centering
\includegraphics[width=.88\columnwidth] {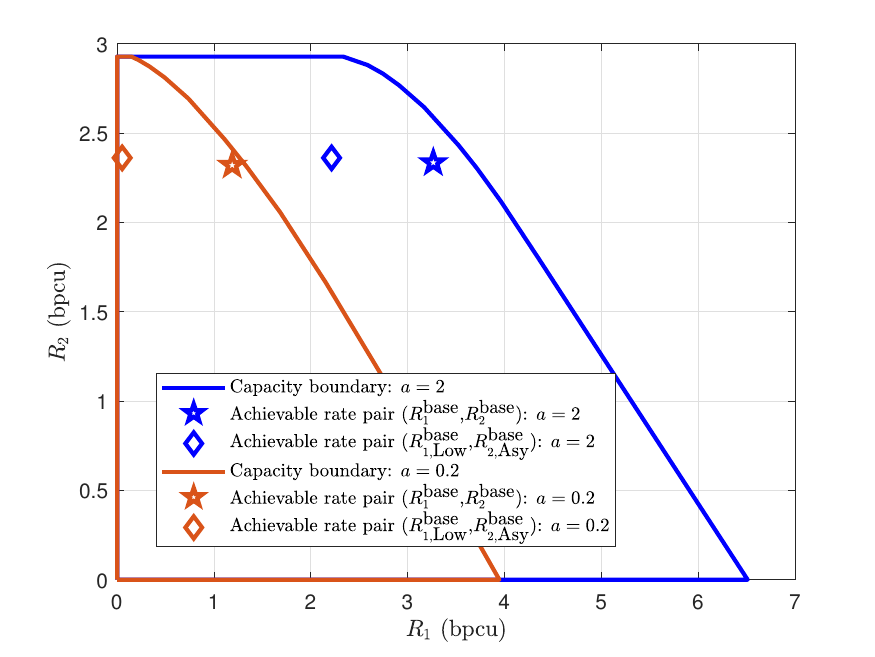}
\caption{Capacity region with $P/\sigma^2 = 10$ dB.}
\label{fig:regiona}
\end{figure}

Fig.~\ref{fig:regiona} illustrates the capacity region for different values of $a$ with the SNR fixed at $P/\sigma^2 = 10$ dB. The parameter $a$ plays a critical role in determining the balance between the additive and multiplicative components of the received signal.
As $a$ increases from $0.2$ to $2$, a substantial expansion is observed along the $R_1$ axis. This phenomenon is expected, as larger $a$ enhances the effective SNR of the direct link, thereby shifting the $R_1$ intercept of the capacity boundary to the right. Notably, when $a=0.2$, the PT's performance highly dependent on the backscatter link $X_2$ to provide additional signaling paths.
Interestingly, the maximum achievable rate for the BD remains relatively invariant to the changes in $a$. This behavior stems from the fact that $R_2$ is primarily governed by the multiplicative term $X_1 X_2$ and the PT's power $P$, rather than the direct link signal. 
With a stronger direct link, the receiver can more effectively distinguish between the additive component $a X_1$ and the multiplicative component $X_1 X_2$. 
This improved allows both users to simultaneously operate near their respective maximum achievable rates, leading to a more rectangular capacity region profile.
Consistent with the previous observations, the baseline achievable rate pairs and their analytical lower bounds remain within the capacity boundary.


\begin{figure}[t]
\centering
\includegraphics[width=.88\columnwidth] {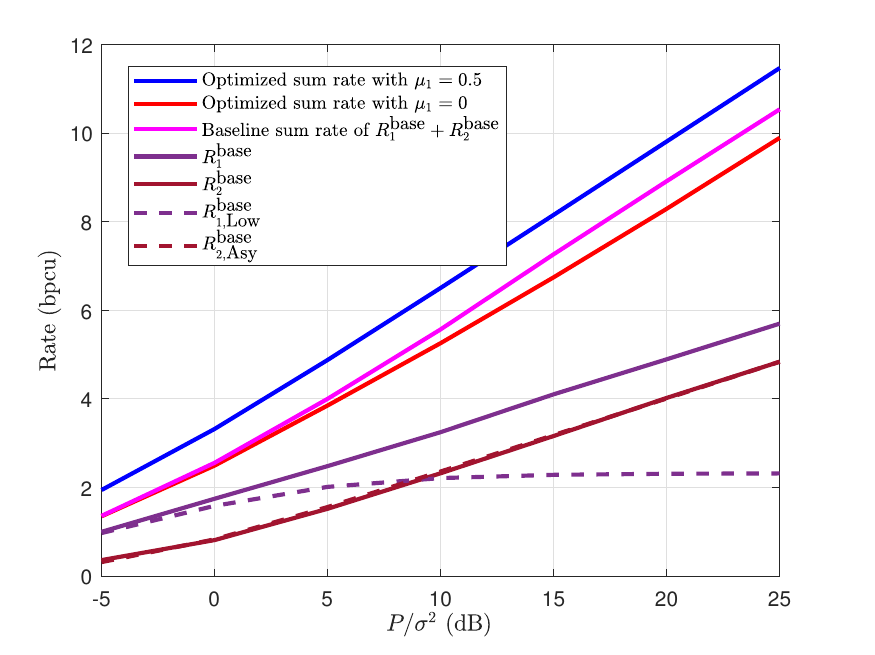}
\caption{Rate versus $P/\sigma^2$ with $a=2$.}
\label{fig:ratesnr}
\vspace{-2em}
\end{figure}

Fig.~\ref{fig:ratesnr} illustrates various rate metrics versus $P/\sigma^2$ with $a=2$. The optimized sum rate with $\mu_1 = 0.5$ consistently outperforms the baseline sum rate, which in turn exceeds the optimized sum-rate for the boundary case where $\mu_1 = 0$.
The secondary rate $R_2^{\text{base}}$ and its asymptotic approximation $R_{2,\text{Asy}}^{\text{base}}$ are nearly indistinguishable across the entire SNR range, validating the tightness of the asymptotic results. 
While the accurate baseline rate $R_1^{\text{base}}$ continues to grow logarithmically with SNR, its lower bound $R_{1,\text{Low}}^{\text{base}}$ exhibits saturation in the high-SNR regime since this bound treats the backscatter term $X_1 X_2$ as interference. 
The fact that the accurate baseline rate $R_1^{\text{base}}$ does not saturate proves that the receiver can effectively exploit the structured multiplicative coupling, ensuring that the PT's rate remains power-limited rather than interference-limited as the transmit SNR increases.



Fig.~\ref{fig:ratea} illustrates the impact of $a$ on the achievable rates and the optimized sum rate of the AM-MAC with $P/\sigma^2 = 10$ dB.
A notable trend is observed in the performance gap between the optimized sum rate and the baseline scheme. 
In the regime where the additive and multiplicative components of the received signal are of comparable magnitude, i.e., $a \approx 1$, this gap reaches its maximum, suggesting that the structural complexity of the coupled interference is most pronounced and that the gains from optimizing discrete mass points are maximal.
Conversely, as $a$ increases beyond $1$, the gap between the optimized maximum sum rate and the baseline sum rate gradually diminishes, with the two curves exhibiting asymptotic convergence. 
This behavior yields a significant practical insight that when the direct-link is sufficiently strong, the conventional baseline signaling scheme becomes near-optimal. 
In such scenarios, the system can bypass the computationally expensive numerical search for optimal discrete mass points. 
Instead, the closed-form approximations and analytical distributions can be used for performance prediction and resource allocation.
Furthermore, both $R_1^{\text{base}}$ and its lower bound $R_{1,\text{Low}}^{\text{base}}$ exhibit logarithmic growth with respect to $a$. This is consistent with the system model where $a$ directly scales the power of the PT's additive component. In contrast, the BD' rate $R_2^{\text{base}}$ remains largely invariant to changes in $a$, reinforcing our earlier observation from Fig.~\ref{fig:regiona} that the BD's performance is primarily governed by the multiplicative coupling rather than the direct-link strength.

\begin{figure}[t]
\centering
\includegraphics[width=.88\columnwidth] {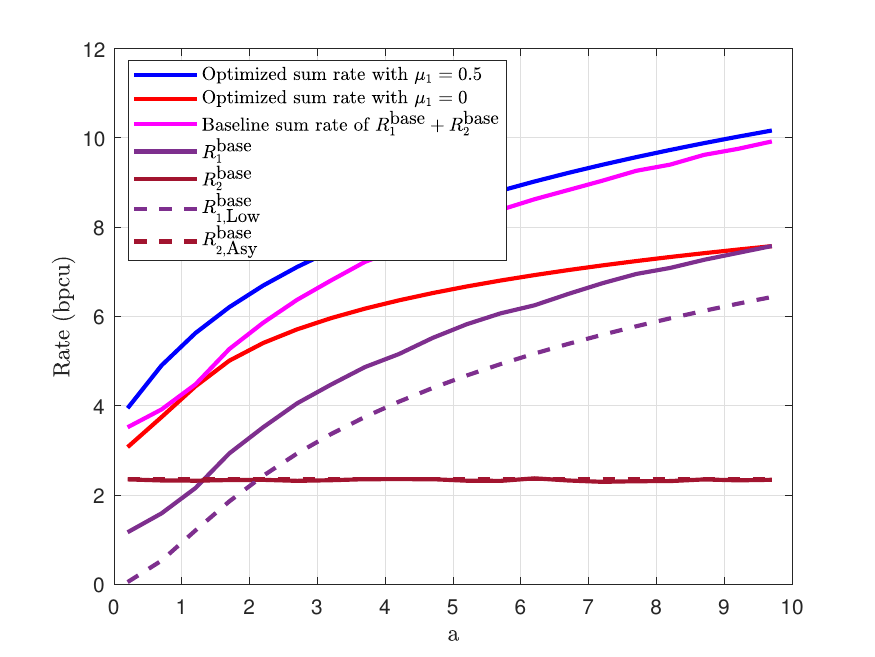}
\caption{Rate versus $a$ with $P/\sigma^2 = 10$ dB.}
\label{fig:ratea}
\vspace{-1em}
\end{figure}

\section {Conclusions}

This paper characterizes the capacity region of a two-user AM-MAC under heterogeneous input constraints, providing an information-theoretic foundation for emerging SR and RIS-enabled information transmission systems. We show that the additive-multiplicative coupling and heterogeneous constraints fundamentally alters the geometry of the capacity region compared to the classical Gaussian MAC.
Our theoretical analysis leads to three key insights. Firstly, the sum-rate capacity is maximized when the BD operates in a purely assistive mode and the PT transmits Gaussian signal. Secondly, to maximize the BD's achievable rate, the PT employs a constant-envelope strategy to provide a stable carrier and the optimal BD signaling follows a concentric-circle structure. Finally, we establish a general discreteness result for the other boundary points where the optimal PT signal consists of a continuous uniform phase and a discrete amplitude, whereas the optimal BD distribution is fully discrete with a finite number of mass points. Numerical results are provided to show the capacity region of AM-MAC under heterogeneous input constraints. 
Future research may generalize this analytical framework to multi-antenna scenarios to facilitate a deeper investigation into the fundamental limits inherent in additive-multiplicative coupling and heterogeneous input constraints.

\appendices

\section{}\label{proof:opt}

The boundary of the capacity region $\mathcal{C}_{\text{AM-MAC}}$ is characterized by solving a weighted sum-rate maximization problem. 
By varying the weight coefficients $\mu_1$ and $\mu_2$, we can trace the entire Pareto frontier of the region. 
In the specific case where $\mu_1 = \mu_2$, the objective function reduces to the total sum-rate $R_1 + R_2$. Under this condition, the maximum sum-rate is determined by the joint mutual information $I(X_1, X_2; Y)$. This is consistent with the capacity results for a standard MAC~\cite{cover1999elements}.

In the regime where the weight $\mu_2$ dominates, i.e., $\mu_1 < \mu_2$, the optimal rate pair $(R_1^*, R_2^*)$ on the Pareto boundary of $\mathcal{C}_{\text{AM-MAC}}$ is defined as the solution to
\begin{equation}
(R_1^*, R_2^*) = \arg \max_{(R_1, R_2) \in \mathcal{C}_{\text{AM-MAC}}} \mu_1 R_1 + \mu_2 R_2.
\end{equation}
By the definition of the capacity region as a closed convex set, there exists a sequence of achievable rate pairs 
${(R_{1,k}, R_{2,k})}_{k=1}^\infty$ converging to $(R_1^*, R_2^*)$. 
Consequently, for any $\varepsilon > 0$, there exists a sufficiently large index $k\varepsilon$ such that the following inequality holds:
\begin{equation} \label{eq:eps_bound}
\mu_1 R_{1,k_\varepsilon} + \mu_2 R_{2,k_\varepsilon} \geq \max_{(R_1, R_2) \in \mathcal{C}_{\text{AM-MAC}}} \mu_1 R_1 + \mu_2 R_2 - \varepsilon.
\end{equation}
Any achievable rate pair $(R_{1,k_\varepsilon}, R_{2,k_\varepsilon})$ within the capacity region can be represented as a convex combination of extreme points, denoted by $(R_{1,k_\varepsilon}, R_{2,k_\varepsilon}) = \sum_{\iota=1}^{\ell} \alpha_\iota (\tilde{R}_{1,\iota}, \tilde{R}_{2,\iota})$, where $\sum_{\iota=1}^{\ell} \alpha_\iota = 1$ and $\alpha_\iota \geq 0$.

Due to the condition $\mu_1 < \mu_2$, the linear objective $\mu_1 R_1 + \mu_2 R_2$ is maximized at the corner point corresponding to the successive interference cancellation (SIC) decoding order $(X_1 \to X_2)$, i.e., $(I(X_1; Y), I(X_2; Y | X_1))$. 
It follows that for each constituent rate pair, we have
\begin{align} \label{eq:sup_bound}
&\mu_1 \tilde{R}_{1,\iota} + \mu_2 \tilde{R}_{2,\iota} \leq \mu_1 I(X_1; Y) + \mu_2 I(X_2; Y | X_1)\nonumber\\
\leq& \max_{X_1:\mathbb{E}[|X_1|^2]\leq P,X_2:|X_2|\leq1}  \mu_1 I(X_1; Y) + \mu_2 I(X_2; Y | X_1).
\end{align}

Combining~\eqref{eq:eps_bound} and~\eqref{eq:sup_bound}, and noting that the inequality holds for an arbitrarily small $\varepsilon > 0$, we arrive at the upper bound:
\begin{align*} 
&\max_{(R_1, R_2) \in \mathcal{C}_{\text{AM-MAC}}} \mu_1 R_1 + \mu_2 R_2 \nonumber\\
\leq &\max_{X_1:\mathbb{E}[|X_1|^2]\leq P,X_2:|X_2|\leq1} \mu_1 I(X_1; Y) + \mu_2 I(X_2; Y | X_1) .
\end{align*}
Conversely, for any fixed input distributions $f_{X_1}$ and $f_{X_2}$, the corner point $(I(X_1; Y), I(X_2; Y | X_1))$ is achievable via SIC, where $X_1$ is decoded first followed by $X_2$. This ensures that the lower bound also holds. Therefore, we conclude that for the regime $\mu_1 < \mu_2$, the weighted sum rate is identically given by:
\begin{align*}
&\max_{(R_1, R_2) \in \mathcal{C}_{\text{AM-MAC}}} \mu_1 R_1 + \mu_2 R_2 \nonumber\\
=& \max_{X_1:\mathbb{E}[|X_1|^2]\leq P,X_2:|X_2|\leq1}  \mu_1 I(X_1; Y) + \mu_2 I(X_2; Y | X_1) .
\end{align*}
The proof for the symmetric case where $\mu_1 > \mu_2$ follows an analogous argument and is thus omitted for brevity. 

\section{}\label{proof:x1concave}

By subtracting the deterministic component $x_1$ from the channel output, the conditional channel from $X_2$ to $Y$ can be reformulated as $\tilde{Y} =  x_1 X_2 + Z$. 
To simplify this into a standard form, we define the SNR as $\gamma \triangleq \frac{|x_1|^2}{\sigma^2}$ and perform phase compensation and noise normalization by setting $\bar{Y} \triangleq \tilde{Y}e^{-j\theta_1}/\sigma$, where $\theta_1$ is the phase of $x_1$. This yields a standard complex Gaussian channel:
\begin{equation}
\bar{Y} = \sqrt{\gamma}X_2 + \bar{Z},
\label{eq:std-gaussian}
\end{equation}
where $\bar{Z} \sim \mathcal{CN}(0, 1)$ represents the normalized AWGN that is independent of $X_2$.

Since the aforementioned transformations are bijective and independent of $X_2$, they are information-preserving. Consequently, the conditional mutual information depends on $x_1$ only through the instantaneous SNR $\gamma$, and can be expressed as
\begin{equation}
I(X_2;Y\mid X_1=x_1)
= I(X_2;\bar{Y})
\triangleq F(\gamma).
\label{eq:I-as-J}
\end{equation}

To characterize the curvature of $F(\gamma)$, we invoke the I-MMSE relationship~\cite{GuoShamaiVerdu2005}, which establishes a fundamental link between information theory and estimation theory. 
Specifically, for the scalar complex Gaussian channel in~\eqref{eq:std-gaussian} with a fixed input distribution of $X_2$, the first derivative of the mutual information (in nats) is given by
\begin{equation}
\frac{\mathrm{d}}{\mathrm{d}\gamma} F(\gamma) = \frac{1}{2} \mathrm{mmse}(\gamma),
\end{equation}
where $\mathrm{mmse}(\gamma) \triangleq \mathbb{E} \left[ |X_2 - \mathbb{E}[X_2| \bar{Y}]|^2 \right]$ denotes the MMSE in estimating $X_2$ from $\bar{Y}$ at SNR $\gamma$.

It has been established in~\cite{GuoShamaiVerdu2005} that $\mathrm{mmse}(\gamma)$ is a differentiable and non-increasing function of $\gamma$. More precisely, its derivative is
\begin{equation}
\frac{\mathrm{d}}{\mathrm{d}\gamma} \mathrm{mmse}(\gamma) = -\mathbb{E} \left[ \left( \mathrm{Var}(X_2 \mid \bar{Y}) \right)^2 \right] \leq 0,
\end{equation}
where $\mathrm{Var}(X_2|\bar{Y})$ represents the conditional variance of $X_2$ given the observation $\bar{Y}$.
By differentiating $F(\gamma)$ twice, we obtain
\begin{align*}
J''(\gamma) &= \frac{1}{2} \frac{\mathrm{d}}{\mathrm{d}\gamma} \mathrm{mmse}(\gamma) = -\frac{1}{2} \mathbb{E} \left[ \left( \mathrm{Var}(X_2 \mid \bar{Y}) \right)^2 \right] \leq 0.
\end{align*}
The non-positive second derivative confirms that $F(\gamma)$ is concave. Since $\gamma$ is linearly proportional to $|x_1|^2$, we conclude that $I(X_2; Y |X_1 = x_1)$ is a concave function of $|x_1|^2$.

\section{}\label{proof:x1phase}
Let $X_1 = r e^{j\theta_1}$, and the received signal can be rewritten as $Y = r(a + X_2)e^{j\theta_1} + Z$.
We first analyze the effects of $\theta_1$ on $h(Y)$.
Defining a new input $\tilde X_1 = r e^{j(\theta_1+\phi)}$ with $\phi \sim \mathcal{U}[-\pi,\pi)$ yields a new phase $\tilde \theta_1 = \theta_1 + \phi$ that is uniform on $[-\pi,\pi)$ and independent of $(r,X_2)$.
Then, the new output satisfies $\tilde Y \triangleq  r(a+X_2) e^{j(\theta_1+\phi)} + Z$.

Since $Z$ is circularly symmetric complex Gaussian, rotation by any phase $\phi$ preserves its distribution, i.e., $Z  \overset{d}{=} e^{j\phi} Z$.
Therefore, the PDF of $\tilde Y$ is an average over all rotations of $Y$, i.e.,
$f_{\tilde Y}(\tilde y) = \int_{-\pi}^{\pi}\frac{1}{2\pi} f_Y(\tilde y e^{-j\phi}) \,d\phi$.
Using the concavity of differential entropy, we have
\begin{align}\label{eq:equality}
    h(\tilde Y) \ge \int_{-\pi}^{\pi} \frac{1}{2\pi}h(Y e^{-j\phi})\,d\phi
           = h(Y),
\end{align}
since entropy is invariant under constant phase rotations.
When $\theta_1$ is uniformly distributed over $[-\pi,\pi)$ and independent of $r$, the distribution of $Y$ is identical to that of $\tilde Y$, and thus equality in~\eqref{eq:equality} holds and the differential entropy $h(Y)$ is maximized.

Next, we will analyze the effects of $\theta_1$ on $h(Y|X_1)$. By letting $\bar Y = e^{-j\theta_1} Y = r(a + X_2) + e^{-j\theta_1}{Z}$, we have $Z  \overset{d}{=} e^{-j\theta_1} Z$ and $h(Y |X_1 = r e^{j\theta_1})= h(\bar Y| X_1 = r e^{j\theta_1})$. Since the conditional distribution of $\bar{Y}$ given $X_1$ is invariant with respect to $\theta_1$, the distribution of $\theta_1$ does not affect $h(Y| X_1)=h(\bar Y| X_1)$.
Combining the above results, we conclude that choosing $\theta_1\sim \mathcal{U}[-\pi,\pi)$ maximizes $h(Y)$ without affecting $h(Y| X_1)$, and hence is optimal.

\section{}\label{proof:concave}

Based on the definition of conditional entropy, $h(Y|X_1 = re^{i\theta})$ is linear with respect to $f_r(r)$. Next, we examine the marginal entropy $h(Y)$. The marginal density $f_Y(y)$ is a linear functional of $f_r(r)$, expressed as
\begin{equation}\label{eq:59}
f_Y(y; f_r) = \int_0^\infty K(y, r) f_r(r) dr.
\end{equation}
For $f_r^\epsilon = \epsilon f_r + (1-\epsilon)f_r^*$ with $\epsilon \in [0,1]$, by linearity of the integral operator, we have
$f_Y(y; f_r^\epsilon) = \epsilon f_Y(y; f_r) + (1-\epsilon) f_Y(y; f_r^*)$. Given that the differential entropy $h(Y)$ is a concave functional of the density $f_Y(y)$~\cite{cover1999elements}, it follows that
\begin{equation*}
h(Y;f_Y(y; f_r^\epsilon))\! \geq \!\epsilon h(Y;f_Y(y; f_r))\! +\! (1-\epsilon) h(Y;f_Y(y; f_r^*)).
\end{equation*}
Consequently, $h(Y)$ is concave in $f_r$, satisfying:
\begin{equation}
h(Y; f_r^\epsilon) \geq \epsilon h(Y; f_r) + (1-\epsilon) h(Y; f_r^*).
\end{equation}

\section{}\label{proof:condition}

We establish the validity of~\eqref{eq:condtion11} in Theorem \ref{theorem:feature} by contradiction. 
Suppose, to the contrary, that \eqref{eq:condtion11} does not hold. Then, there must exist some $\tilde{r}$ such that the following inequality is satisfied:
\begin{equation} \label{eq:contradiction_hypo}
\omega_1(\tilde{r}; f_r^{*}) > J_1 - \left( \frac{\mu_2}{\mu_1} - 1 \right) h(Y | \tilde{r}) + \frac{\lambda}{\mu_1} \tilde{r}^2.
\end{equation}
Recall that \eqref{eq:cond11} is required to hold for all distributions $f_r \in \mathcal{F}_r$. 
By choosing a specific distribution $f_r = \delta(r - \tilde{r})$, where $\delta(\cdot)$ denotes the Dirac delta function, and substituting it into the integral form of~\eqref{eq:cond11}, we obtain
\begin{equation}
\int f_r \left( \omega_1(r; f_r^{*}) + \left( \frac{\mu_2}{\mu_1} - 1 \right) h(Y | r) - \frac{\lambda}{\mu_1} r^2 \right) dr > J_1.
\end{equation}
This result directly contradicts the foundational condition established in~\eqref{eq:cond11}, which asserts that the integral must not exceed $J_1$. 
Thus, \eqref{eq:condtion11} is valid.

Next, we establish the validity of~\eqref{eq:condtion12} in Theorem~\ref{theorem:feature}, again employing a proof by contradiction.
Assume that there exists a point $\tilde{r}$ in the support $\mathcal S_{f_r^*}$ such that the following strict inequality holds:
\begin{equation} \label{eq:hypo_lower}
\omega_1(\tilde{r}; f_r^{*}) < J_1 - \left( \frac{\mu_2}{\mu_1} - 1 \right) h(Y | \tilde{r}) + \frac{\lambda}{\mu_1} \tilde{r}^2.
\end{equation}
By the definition of a point of increase, any neighborhood $\tilde{\mathcal{S}} \subset \mathcal S_{f_r^*}$ centered at $\tilde{r}$ must have a strictly positive probability measure, denoted by $\int_{\tilde{\mathcal{S}}} f_r^* dr = \varrho > 0$. Given the continuity of the functional components, there exists a sufficiently small neighborhood $\tilde{\mathcal{S}}$ such that~\eqref{eq:hypo_lower} holds for all $r \in \tilde{\mathcal{S}}$.
Under this assumption, we evaluate $J_1$ by partitioning the support into $\tilde{\mathcal{S}}$ and its complement $\mathcal S_{f_r^*} \setminus \tilde{\mathcal{S}}$:
\begin{align*}
J_1 &= \int_{\mathcal S_{f_r^*}} f_r^{*} \left( \omega_1(r; f_r^{*}) + \left( \frac{\mu_2}{\mu_1} - 1 \right) h(Y | r) - \frac{\lambda}{\mu_1} r^2 \right) dr \\
&= \int_{\tilde{\mathcal{S}}} f_r^{*} \left( \omega_1(r; f_r^{*}) + \left( \frac{\mu_2}{\mu_1} - 1 \right) h(Y | r) - \frac{\lambda}{\mu_1} r^2\right) dr \\
&+ \int_{\mathcal S_{f_r^*} \setminus \tilde{\mathcal{S}}} f_r^{*} \left( \omega_1(r; f_r^{*}) + \left( \frac{\mu_2}{\mu_1} - 1 \right) h(Y | r) - \frac{\lambda}{\mu_1} r^2 \right) dr \\
&< \int_{\tilde{\mathcal{S}}} f_r^{*} J_1 dr + \int_{\mathcal S_{f_r^*} \setminus \tilde{\mathcal{S}}} f_r^{*} J_1 dr \\
&= \varrho J_1 + (1 - \varrho) J_1 = J_1.
\end{align*}
The resulting strict inequality $J_1 < J_1$ is a clear contradiction. Thus, we conclude that the condition in~\eqref{eq:condtion12} must hold for all $r$ in the support $\mathcal S_{f_r^*}$.

\section{}\label{proof:real}

To prove that $A(r)$ is real analytic in $r$, it suffices to show that $A(r)$ can be represented by a power series with a non-zero radius of convergence.
Recall that $A(r) =  -\int_{\mathbb C} K(y, r)\log f_y(y) dy$, where $K(y, r)\! \triangleq\!\frac{1}{2\pi}\int_{\mathbb C}\!\int_{-\pi}^{\pi}\! \frac{f_{X_2}(x_2)}{\pi\sigma^2} \exp(-\frac{|y - (a+x_2)re^{j\theta_1}|^2}{\sigma^2}) d\theta_1dx_2$.
Let $\alpha(y, X_2, \theta_1) \triangleq \frac{2\text{Re}[y^H(a+X_2)e^{j\theta_1}]}{\sigma^2}$ and $\beta(X_2) \triangleq \frac{|a+X_2|^2}{\sigma^2}$, yielding
\begin{align*}
K(y, r)=\int_{\mathbb C}\!\int_{-\pi}^{\pi}\! \frac{f_{X_2}(x_2)}{2\pi^2\sigma^2} e^{-\frac{|y|^2}{\sigma^2}} \exp(\alpha r - \beta r^2) d\theta_1dx_2.
\end{align*}
Since $\exp(x)$ is an entire function, the term $\exp(\alpha r - \beta r^2)$ can be expanded into a power series $\sum_{n=0}^{\infty} d_n(\alpha, \beta) r^n$ that converges absolutely for all $r \in \mathbb{R}$.

For each fixed $y$, the coefficients $d_n(\alpha, \beta)$ are bounded because $X_2$ and $\theta_1$ reside in compact sets. Consequently, the series $\sum_{n=0}^{\infty} d_n r^n$ converges uniformly with respect to $(X_2, \theta_1)$. By the Dominated Convergence Theorem, we can interchange the integral and the summation
\begin{align}\label{eq:Kyr}
K(y,r) &= \sum_{n=0}^{\infty} \left[\int_{\mathbb C}\!\int_{-\pi}^{\pi}\frac{f_{X_2}(x_2)}{2\pi^2\sigma^2} e^{-\frac{|y|^2}{\sigma^2}} d_n(\alpha, \beta) d\theta_1 dx_2 \right] r^n \nonumber\\
& \triangleq \sum_{n=0}^{\infty} K_n(y) r^n.
\end{align}

The final step requires interchanging the integral over $y$ with the summation over $n$, yielding
\begin{align*}
A(r) = -\int_{\mathbb C} \left( \sum_{n=0}^{\infty} K_n(y) r^n \right) \log f_Y(y) dy.
\end{align*}
The interchange is valid if the series $\sum_{n=0}^{\infty} |K_n(y) \log f_Y(y) r^n|$ is integrable. Given that $-\log f_Y(y)$ grows at most quadratically as $|y| \to \infty$, as proved in Lemma \ref{lem:gaussian-lower-envelope}, and $K_n(y)$ contains the Gaussian factor $e^{-{|y|^2}/{\sigma^2}}$, the integrand is dominated by an integrable Gaussian-type function. Thus, $A(r)$ can be expressed as a power series $A(r) = \sum_{n=0}^\infty A_n r^n$, where $A_n \triangleq -\int_{\mathbb{C}} K_n(y) \log f_Y(y) dy$. Since this series converges for all $r$ within its radius of convergence, $A(r)$ is real analytic.

Similarly, to prove that $B(r)$ is real analytic, it suffices to show that $h(Y|r)$ admits a power series expansion in $r$ with a non-zero radius of convergence. Recall that $h(Y|r) = -\int K(y,r) \log K(y,r) dy$.
According to~\eqref{eq:Kyr}, for a fixed $y$, as $r \to 0$, the PDF $K(y,r)$ approaches $K_0(y)$, where $K_0(y)\triangleq\frac{1}{\pi\sigma^2} e^{-|y|^2/\sigma^2} > 0$ for all $y \in \mathbb{C}$, yielding
\begin{align}\label{eq:log}
\log K(y,r) = \log K_0(y) + \log\left( 1 + \sum_{n=1}^{\infty} \frac{K_n(y)}{K_0(y)} r^n \right).
\end{align}
Using the Taylor expansion $\log(1+x) = \sum_{k=1}^\infty \frac{(-1)^{k+1}}{k} x^k$, which converges absolutely for $|x| < 1$, we can expand the second term in~\eqref{eq:log} into a power series $\sum_{m=1}^\infty h_m(y) r^m$. This expansion is valid for $|r| < \rho_y$. Since $V_0$ is bounded, one can show that there exists a uniform $\rho > 0$ such that the series converges for all $y$ within any compact subset of $\mathbb{C}$.

The integrand in $h(Y|r)$ is the product of two power series, which is given by
$\left( \sum_{n=0}^{\infty} K_n(y) r^n \right) \left( \log K_0(y) + \sum_{m=1}^{\infty} h_m(y) r^m \right) = \sum_{k=0}^{\infty} \phi_k(y) r^k,
$
where $\phi_k(y) \triangleq K_k(y) \log K_0(y) + \sum_{n+m=k, m \geq 1} K_n(y) h_m(y)$.

To conclude that $h(Y|r) = \sum_{k=0}^\infty \left[ -\int \phi_k(y) dy \right] r^k$ is analytic, we must justify the interchange of the integral and the infinite sum.
The term $K_n(y)$ decays as $O(e^{-|y|^2/\sigma^2})$ and $\log K_0(y)$ grows as $O(|y|^2)$. Thus, each $\phi_k(y)$ is integrable.
Since the original kernel is an entire function of $r$ and the logarithm is analytic away from zero, the composite function is analytic. 
The series $\sum \phi_k(y) r^k$ is dominated by an integrable Gaussian-type function for sufficiently small $r$.
By the Dominated Convergence Theorem, the interchange is valid, confirming that $h(Y|r)$, and thus $B(r)$, is real analytic at $r=0$.

\section{}\label{proof:x1}

To analyze the properties of the optimal input distribution, we decompose the KKT condition \eqref{eq:condtion12} by defining $A(r)\triangleq \omega_1(r;f_r^{*})$ and $B(r)\triangleq \left(\frac{\mu_2}{\mu_1}-1\right)h(Y|r)$. Consequently, for any $r \in \mathcal{S}_{f_r^*}$, the following stationarity condition must hold
\begin{align}
A(r)+B(r)-\frac{\lambda}{\mu_1} r^2-J_1 = 0,  r\in\mathcal S_{f_r^*}.
\end{align}
The term $A(r)$ represents the marginal cross-entropy $A(r) = -\int K(y,r)\log f_Y(y)dy = -\mathbb E [\log f_Y(Y_r)]$,
where $Y_r = (a+X_2)r e^{j\theta_1} + Z$ denotes the conditional received signal given the PT input amplitude $|X_1|=r$. Here, the expectation is taken over the joint distribution of $(\theta_1,X_2, Z)$.

\begin{lemma}\label{lem:gaussian-lower-envelope}
Suppose that $X_1$ possesses a finite exponential moment, i.e., there exists a sufficiently small $\alpha > 0$ such that $\mathbb{E}[e^{\alpha |X_1|^2}] < \infty$. Then, the marginal output PDF $f_Y(y)$ satisfies a Gaussian-type lower envelope:
\begin{equation}\label{eq:GaussianEnvelope}
-\log f_Y(y) \geq c_1 |y|^2 - c_2, 
\end{equation}
where $c_1 > 0$ and $c_2 \in \mathbb{R}$ are constants depending on $\sigma^2$ but independent of the specific realization of $y$.
\end{lemma}
\begin{IEEEproof}
We define $X \triangleq (1+aX_2)X_1$. The marginal output PDF $f_Y(y)$ is the convolution of the input distribution and the AWGN noise, which can be expressed as
\begin{align}
f_Y(y) = \int_{\mathbb{C}} f_Z(y-x) f_X(x) dx = \mathbb{E}[f_Z(y-X)],
\end{align}
where $f_Z(z) = \frac{1}{\pi\sigma^2} \exp(-|z|^2/\sigma^2)$. By applying the triangle inequality $|y-x|^2 \le (|y| + |x|)^2$, we obtain the following bound for the noise PDF
\begin{align*}
f_Z(y-x)
&=
\frac{1}{\pi\sigma^2}
\exp\!\left(-\frac{|y|^2+|x|^2-2\Re (yx^{H})}{\sigma^2}\right)\\
&\leq
\frac{1}{\pi\sigma^2}
\exp\!\left(-\frac{|y|^2+|x|^2-2|y||x|)}{\sigma^2}\right).
\end{align*}
Substituting this into the expression for $f_Y(y)$ yields
\begin{align}\label{eq:f_Y_bound}
f_Y(y)&\le
\frac{e^{-|y|^2/\sigma^2}}{\pi\sigma^2}
\mathbb{E}\!\left[
\exp\!\left(\frac{2|y||X|-|X|^2}{\sigma^2}\right)
\right].
\end{align}

To bound the expectation term, we invoke Young's inequality in the form $2|y||X| \leq \frac{|y|^2}{\epsilon} + \epsilon |X|^2$ for any $\epsilon > 0$, yielding
\begin{align}
\exp\!\left(\frac{2|y||X|}{\sigma^2}-\frac{|X|^2}{\sigma^2}\right) \leq 
\exp\!\left(\frac{|y|^2}{\epsilon\sigma^2}\right)\exp\!\left(-\frac{(1-\epsilon)|X|^2}{\sigma^2}\right)
\end{align}
Substituting this into~\eqref{eq:f_Y_bound}, the PDF is bounded by
\begin{align}
f_Y(y)
\le
\frac{1}{\pi\sigma^2}\exp\!\left(-\frac{|y|^2}{\sigma^2}+\frac{|y|^2}{\epsilon\sigma^2}\right)\mathbb{E}\!\left[\exp\!\left(-\frac{(1-\epsilon)|X|^2}{\sigma^2}\right)\right].
\end{align}

By taking $c_1 = \frac{\epsilon-1}{\epsilon\sigma^2}$ and $\beta = \epsilon c_1$, we have 
\begin{align}
f_Y(y)\le\frac{1}{\pi\sigma^2}\exp\!\left(-c_1|y|^2\right)\mathbb{E}\!\left[\exp\!\left(\beta|X|^2\right)\right].
\end{align}
Since the above inequality holds for all $\epsilon \ge 0$, we take $ \epsilon \ge 1$ to make sure $c_1>0$ and $\beta>0$.

Given the peak amplitude constraint $|X_2|\leq 1$, it follows that $|X| \leq (1+a)|X_1|$. Thus, $|X|^2 \leq (1+a)^2 |X_1|^2$. 
By the lemma's assumption, $X_1$ possesses a finite exponential moment $\mathbb{E}[e^{\alpha |X_1|^2}] < \infty$, yielding $\mathbb{E}\!\left[e^{\beta|X|^2}\right]$ is finite if $\beta (1+|a|)^2 \leq \alpha$.
By selecting $\epsilon$ sufficiently close to $1$ from above, the constant $\beta$ can be made arbitrarily small to satisfy the condition $\beta(1+|a|)^2 \leq \alpha$. Under this condition, the expectation $C \triangleq \frac{1}{\pi\sigma^2} \mathbb{E}[e^{\beta |X|^2}]$ is guaranteed to be finite. 
Consequently, we obtain the universal upper bound $f_Y(y) \leq C \exp(-c_1 |y|^2)$. 
Taking the negative logarithm on both sides yields the desired lower envelope 
\begin{align}
-\log f_Y(y) \geq c_1 |y|^2 - c_2,
\end{align}
where $c_2 \triangleq \log C$. This concludes the proof.
\end{IEEEproof}

Using Lemma \ref{lem:gaussian-lower-envelope}, we can bound $A(r)$ as $A(r)= -\mathbb{E}[\log f_Y(Y_r)]\geq c_1\,\mathbb{E}|Y_r|^2 - c_2$, where the constants $c_1 > 0$ and $c_2 \in \mathbb{R}$ are independent of $r$. Substituting the second moment $\mathbb{E}[|Y_r|^2] = \mathbb{E}[|a+X_2|^2]r^2 + \sigma^2$ into the inequality yields
\begin{align} \label{eq:lowerAr}
A(r) \geq c_1 \mathbb{E}[|a+X_2|^2] r^2 + c_1 \sigma^2 - c_2.
\end{align}
This indicates that $A(r)$ grows at least quadratically with $r$.

We now analyze the asymptotic behavior of $B(r)$. By scaling the received signal as $Y_r = r [ (a+X_2)e^{j\theta_1} + \frac{Z}{r} ]$, we can write
\begin{align*}
B(r) & = \left(\frac{\mu_2}{\mu_1}-1\right)h(Y|r) \overset{(a)}{=} \left(\frac{\mu_2}{\mu_1}-1\right)(h(V_r)+2\log r),
\end{align*}
where $V_r \triangleq V_0 + \frac{Z}{r}$ and $V_0 \triangleq (a+X_2)e^{j\theta_1}$. Equality $(a)$ follows from the scaling property of differential entropy in the complex plane, i.e., $h(\alpha \mathbf{X}) = h(\mathbf{X}) + 2\log |\alpha|$.

As $r \to \infty$, the term $Z/r$ vanishes for any realization of $Z$. To establish the convergence of the corresponding densities, we examine the PDF of $V_r$, denoted by $f_{V_r}(\cdot)$. Due to the additive noise structure, $f_{V_r}$ is given by the convolution
\begin{equation}\label{eq:conv}
f_{V_r}(v) = \int_{\mathbb{C}} f_{V_0}(u) f_{Z/r}(v-u) du = (f_{V_0} * f_{Z/r})(v),
\end{equation}
where $f_{Z/r}(z) = \frac{r^2}{\pi\sigma^2} \exp\left(-\frac{r^2|z|^2}{\sigma^2}\right)$. This Gaussian kernel forms an approximate identity as $r \to \infty$~\cite{folland1999real}. Consequently, at every Lebesgue point of $f_{V_0}$, we have $f_{V_r}(v) \to f_{V_0}(v)$ pointwise almost everywhere. Since $f_{V_r}$ and $f_{V_0}$ are valid PDFs, implying $\|f_{V_r}\|_1 = \|f_{V_0}\|_1 = 1$, it follows from Scheffé's Theorem that the densities converge in $L^1$ norm:
\begin{equation}
\lim_{r \to \infty} \int_{\mathbb{C}} |f_{V_r}(v) - f_{V_0}(v)| dv = 0.
\end{equation}

Combining the aforementioned arguments, we conclude that $f_{V_r}$ converges to $f_{V_0}$ both pointwise almost everywhere and in $L^1$ norm as $r \to \infty$. This convergence justifies the exchange of the limit and the integral in the entropy functional. Specifically, since $f_{V_r} \to f_{V_0}$ a.e., the term $f_{V_r} \log f_{V_r}$ also converges to $f_{V_0} \log f_{V_0}$ a.e. Furthermore, as shown in Appendix~\ref{proof:upper}, the sequence $|f_{V_r} \log f_{V_r}|$ is dominated by an integrable function, allowing the application of the dominated convergence Theorem~\cite{folland1999real}:
\begin{equation*}
-\int_{\mathbb{C}} f_{V_r}(v) \log f_{V_r}(v) dv \xrightarrow[r \to \infty]{} -\int_{\mathbb{C}} f_{V_0}(v) \log f_{V_0}(v) dv.
\end{equation*}
Consequently, the continuity of differential entropy under $L^1$ convergence and uniform power constraints implies $h(V_r) \to h(V_0) = h((a+X_2)e^{j\theta_1})$. 
Thus, the asymptotic behavior of $h(Y_r)$ as $r \to \infty$ is characterized by:
\begin{equation}\label{eq:asy}
B(r) = \left(\frac{\mu_2}{\mu_1}-1\right)(2\log r + C_0) + o(1),
\end{equation}
where the constant $C_0 \triangleq h((a+X_2)e^{j\theta_1})$ represents the differential entropy of the noise-free normalized signal, and $o(1)$ denotes a term that vanishes as $r \to \infty$.

Combining the lower bound \eqref{eq:lowerAr} on $A(r)$ with the asymptotic expansion~\eqref{eq:asy} of $B(r)$, we obtain the following characterization for $r \in \mathcal{S}_{f_r^*}$
\begin{align} \label{eq:asy_refined}
 &\psi(r) \triangleq A(r)+B(r)-\frac{\lambda}{\mu_1} r^2-J_1\!\geq\! \left(\!c_1\mathbb{E}|a+X_2|^2-\frac{\lambda}{\mu_1}\right)r^2 \nonumber\\
 &+ \left(\frac{\mu_2}{\mu_1}-1\right)(2\log r + C_0) + \Delta,
\end{align}
where $\Delta \triangleq c_1\sigma^2 - c_2 - J_1$.

From Appendix~\ref{proof:real}, the functions $A(r)$, $B(r)$, and $r^2$ are real-analytic functions of $r$, and thus the function $A(r)+B(r)-\frac{\lambda}{\mu_1} r^2$ is real-analytic on $(0, \infty)$.
Suppose the optimal amplitude distribution $f_r^*$ possesses a continuous component on a set $\mathcal{S}$ with positive Lebesgue measure. The KKT conditions then necessitate that $\psi(r) = 0$ for all $r \in \mathcal{S}$. 
Since any set of positive measure in $(0, \infty)$ contains an accumulation point, the Identity Theorem for real-analytic functions~\cite{folland1999real} implies that $\psi(r)$ must be identically zero over the entire domain $(0, \infty)$.

However, this leads to a direct contradiction with the asymptotic behavior established in~\eqref{eq:asy_refined}. Specifically, for $\frac{\mu_2}{\mu_1} \neq 1$, the function $\psi(r)$ is dominated by either the quadratic term $r^2$ or the logarithmic term $\log r$ as $r \to \infty$. Such a function cannot vanish identically on $(0, \infty)$ unless all coefficients are zero, which is not the case here. Consequently, the support $\mathcal{S}_{f_r^*}$ cannot contain any set of positive measure, forcing the optimal distribution $f_r^*$ to be discrete.

Furthermore, since $\lim_{r\to\infty} |\psi(r)| = \infty$, the zeros of $\psi(r)$ are confined to a compact subset of $(0, \infty)$. Given that the zeros of a non-identically zero analytic function cannot have an accumulation point within its domain, we conclude that the support $\mathcal{S}_{f_r^*}$ is a finite set of discrete points.

\begin{remark}
The above analysis requires that the input distribution of $X_1$ possesses a finite exponential moment. When $X_1$ fails to have such a moment, the KKT equation also cannot be satisfied by any continuous input distribution. Specifically, let $X_{1}$ satisfy $\mathbb{E}[|X_{1}|^{2}]<P$ but have no finite exponential moment, i.e.,
$\mathbb{E}[ e^{\alpha |X_{1}|^{2}} ] = \infty, \forall \alpha>0$.
Hence the Gaussian domination argument fails and the
output density inherits the heavy tail of $X_{1}$, implying that $-\log f_{Y}(y)$ grows strictly slower than any quadratic function of $|y|$.
Therefore $A(r)$ grows strictly slower than $r^{2}$. This contradicts the KKT equation, whose right--hand side contains a term $\frac{\lambda}{\mu_{1}} r^{2}$ and must grow quadratically. Thus, no continuous input PDF can satisfy the KKT condition.
\end{remark}

\section{}\label{proof:upper}

To apply the Dominated Convergence Theorem, we must identify an integrable function $g(v)$ such that $|f_{V_r}(v) \log f_{V_r}(v)| \leq g(v)$ for all $r \geq r_0 > 0$.
From~\eqref{eq:conv}, we observe that $f_{V_r}(v)$ exhibits Gaussian tail decay. Using Lemma~\ref{lem:gaussian-lower-envelope} and the property that $V_0$ is bounded, there exist constants $A, B > 0$ such that for sufficiently large $|v|$:
\begin{align}
  f_{V_r}(v) \leq A \exp(-B |v|^2).
\end{align}
This ensures that $f_{V_r} \log f_{V_r}$ decays faster than any polynomial at infinity.
Since $f_{V_r}(v)$ exhibits Gaussian tail decay of the order $e^{-B|v|^2}$, the term $|f_{V_r} \log f_{V_r}|$ is asymptotically dominated by $|v|^2 e^{-B|v|^2}$. Due to the rapid decay of the exponential function relative to the quadratic term, there exist constants $\tilde{A}$ and $\tilde{B}$ such that $|f_{V_r} \log f_{V_r}| \leq \tilde{A} e^{-\tilde{B}|v|^2}$, which is integrable.
Since the noise variance $\sigma^2/r^2$ is strictly positive for any finite $r$, the PDF $f_{V_r}(v)$ is uniformly bounded, i.e., $f_{V_r}(v) \leq \frac{r^2}{\pi \sigma^2} \leq M$.

By combining the uniform boundedness for small $|v|$ and the Gaussian tail decay for large $|v|$, we can construct
\begin{align}
g(v) = 
\begin{cases} 
M, & |v| \leq R \\ \tilde A e^{-\tilde B |v|^2}, & |v| > R 
\end{cases}
\end{align}
where $M, \tilde A, \tilde B$ and the radius $R$ are constants independent of $r$. We construct $g(v)$ to be uniformly bounded in a compact region 
and exhibit Gaussian decay in the tail, ensuring integrability, i.e., $\int_{\mathbb{C}} g(v) dv < \infty$, and thus the sequence $|f_{V_r} \log f_{V_r}|$ is indeed dominated by an integrable function.
The existence of such an integrable dominator $g(v)$ justifies the exchange of the limit and the integral.

\section{}\label{proof:x2concave}

To prove the concavity of $h(Y|X_1)$ with respect to the distribution $f_{X_2}$, we first observe that the conditional PDF $f_{Y|X_1}(y|x_1)$ is a linear functional of $f_{X_2}$. 
Specifically, given $p(y|x_1, x_2) = \frac{1}{\pi\sigma^2} \exp(-|y - (a+x_2)x_1|^2/\sigma^2)$, we have
\begin{equation}
f_{Y|X_1}(y|x_1) = \int_{\mathbb{C}} p(y|x_1, x_2) f_{X_2}(x_2) dx_2.
\end{equation}
Let $f_{X_2}^\epsilon = \epsilon f_{X_2} + (1-\epsilon) f_{X_2}^*$ be a convex combination of two input distributions. The corresponding output density satisfies
\begin{equation}
f_{Y|X_1}^\epsilon(y|x_1) = \epsilon f_{Y|X_1}(y|x_1) + (1-\epsilon) f_{Y|X_1}(y|x_1).
\end{equation}
Recall that the differential entropy $h(f)$ is a concave functional of the density $f$. Thus, for any fixed $x_1 = re^{j\theta_1}$
\begin{align}
&h(Y|X_1=x_1; f_{X_2}^\epsilon)= h(f_{Y|X_1}^\epsilon) \nonumber\\
\geq &\epsilon h(f_{Y|X_1}) + (1-\epsilon) h(f_{Y|X_1}) \nonumber\\
= &\epsilon h(Y|X_1=x_1; f_{X_2}) + (1-\epsilon) h(Y|X_1=x_1; f_{X_2}).
\end{align}
Taking the expectation over $X_1$ preserves the inequality due to the linearity of the expectation operator
\begin{align}
&h(Y|X_1; f_{X_2}^\epsilon)= \mathbb{E}_{X_1} [h(Y|X_1=x_1; f_{X_2}^\epsilon)] \nonumber\\
\geq& \epsilon h(Y|X_1; f_{X_2}) + (1-\epsilon) h(Y|X_1; f_{X_2}).
\end{align}
This concludes that $h(Y|X_1)$ is concave in $f_{X_2}$. The proof for $h(Y)$ follows a similar logic and is omitted here.

\section{}\label{proof:x2}

To establish the optimality and structure of the BD input distribution, we analyze the analytical properties of the integrand in $\omega_2(X_2; f_{X_2})$. 
Let $X_2 = u + jv$, where $u, v \in \mathbb{R}$. 
The integrand is reformulated as
\begin{equation}\label{eq:integrand_refined_v2}
g(x_1, y; u, v) \triangleq f_{X_1}(x_1) f_Z(y - (a+u+jv)x_1) \kappa(x_1, y),
\end{equation}
where $\kappa(x_1, y) \triangleq \mu_1 \log f_Y(y) + (\mu_2 - \mu_1)\log f_{Y|X_1}(y|x_1)$ is independent of the realization of $X_2$. 
Given the discrete nature of $f_{X_1}$ and the Gaussian noise PDF, for each $x_1 = r_k e^{j\theta_1}$, we have
\begin{align*}
g(x_1, y; u, v) = \frac{p_k}{2\pi^2\sigma^2} \exp \left( E(r_k e^{j\theta_1}, y; u, v) \right) \kappa(r_k e^{j\theta_1}, y),
\end{align*}
where the exponent $E(r_k e^{j\theta_1}, y; u, v) \triangleq -\frac{1}{\sigma^2}|y - (a+u+jv)r_k e^{j\theta_1}|^2$.

Since $E$ is a quadratic polynomial in $u$ and $v$, it is inherently real-analytic on $\mathbb{R}^2$. As the exponential mapping is an entire function, the composition $\exp(E)$ remains real-analytic. Consequently, for any fixed $(x_1, y)$, $g(x_1, y; u, v)$ admits a convergent Taylor series expansion in the neighborhood of any $(u, v) \in \mathbb{R}^2$.

To justify the interchange of the integral and the derivative operators, we propose the following
\begin{corollary}\label{cor:derivative_domination}
For any multi-index $\alpha = (\alpha_1, \alpha_2) \in \mathbb{N}^2$, there exists an integrable function $M_\alpha(x_1, y)$ such that
\begin{equation}
\left| \frac{\partial^{|\alpha|} g}{\partial u^{\alpha_1} \partial v^{\alpha_2}}(x_1, y; u, v) \right| \leq M_\alpha(x_1, y)
\end{equation}
holds uniformly for all $(u, v)$ within the unit disk $\mathcal{D} = \{u^2 + v^2 \leq 1\}$. 
\end{corollary}
\begin{IEEEproof}
Since $\kappa(x_1, y)$ does not depend on $(u, v)$, by using the product rule and chain rule, we have
\begin{equation}
\label{eq:derivative_structure}
\frac{\partial^{|\alpha|} g}{\partial u^{\alpha_1} \partial v^{\alpha_2}} =  \frac{p_k\kappa(x_1, y) }{2\pi^2\sigma^2} \frac{\partial^{|\alpha|}\left[\exp(E(r_k e^{j\theta_1}, y; u, v))\right]}{\partial u^{\alpha_1} \partial v^{\alpha_2}}.
\end{equation}

By the multivariate Fa\`a di Bruno formula, for any multi-index $\alpha=(\alpha_1,\alpha_2)$, the derivative of the exponential can be written as
\begin{equation}\label{eq:faa_di_bruno}
\partial^{|\alpha|} \exp(E(u,v))
= \exp(E(u,v))\,\mathcal{B}_{\alpha}\big(\{\partial^{|\beta|} E(u,v)\}_{1\le|\beta|\le|\alpha|}\big),
\end{equation}
where $\mathcal{B}_{\alpha}(\cdot)$ denotes a multivariate Bell polynomial in the partial derivatives of $E$.

From equations \eqref{eq:derivative_structure} and \eqref{eq:faa_di_bruno}, we have
\begin{align}
&\left|\frac{\partial^{|\alpha|} g}{\partial u^{\alpha_1}\partial v^{\alpha_2}}\right|
\le \frac{p_k|\kappa(x_1,y)|}{2\pi^2\sigma^2}
\exp(E)\times \nonumber\\
&\qquad\qquad\qquad\qquad\qquad\qquad\qquad\Big|\mathcal B_\alpha\big(\{\partial^{|\beta|} E\}_{1\le|\beta|\le|\alpha|}\big)\Big|\nonumber\\
&\overset{(a)}{\leq} \frac{c_1 p_k|\kappa(x_1,y)|}{2\pi^2\sigma^2}
\exp\!\left(-\frac{|y|^2}{2\sigma^2}\right)\,
\Big|\mathcal B_\alpha\big(\{\partial^{|\beta|} E\}_{1\le|\beta|\le|\alpha|}\big)\Big|,
\label{eq:95b}
\end{align}
where (a) follows from Lemma~\ref{lem:gaussian-lower-envelope} in Appendix~\ref{proof:x1}.
Since $E(u,v)$ is quadratic in $(u,v)$, we have $\partial^{|\beta|} E(u,v)\equiv 0$ for all $|\beta|\ge 3$.
Hence, $\mathcal{B}_{\alpha}$ depends only on $\partial E$ and $\partial^2 E$, and $\partial^{|\alpha|} \exp(E)$
can be upper-bounded by $\exp(E)$ times a polynomial in $(u,v)$ of degree at most $2|\alpha|$.
Accordingly, there exists a constant $c_\alpha>0$ and a polynomial $P_\alpha(\cdot,\cdot)$ such that
\begin{align*}
&\Big|B_\alpha\big(\{\partial^{|\beta|} E\}_{1\le|\beta|\le|\alpha|}\big)\Big|\\
\leq& c_\alpha\Big(1+\max_{|\beta|=1}|\partial^{|\beta|} E|^{|\alpha|}
+\max_{|\beta|=2}|\partial^{|\beta|} E|^{|\alpha|}\Big)\\
\leq& c_\alpha{P}_\alpha(r_k e^{j\theta_1}, y),
\end{align*}
where ${P}_\alpha(r_k e^{j\theta_1}, y)$ is a polynomial in $(x_1, y)$ and the polynomial dependence on $(u,v)$ can be absorbed into $c_\alpha$.
Therefore,
\begin{equation*}
M_\alpha(x_1, y) \triangleq \frac{c_1 c_\alpha p_k|\kappa(x_1, y)|}{2\pi^2\sigma^2}  {P}_\alpha(r_k e^{j\theta_1}, y) \exp\left(- \frac{|y|^2}{2\sigma^2}\right).
\end{equation*}

Since $|x_1|=r_k$ is fixed, $P_\alpha(r_ke^{j\theta_1},y)$ is a polynomial in $y$ of degree at most $2|\alpha|$, yielding 
$|P_\alpha(r_k e^{j\theta_1}, y)| \leq A^{|\alpha|} (1 + |y|^{2|\alpha|})$, where $A$ accounts for the number of derivative combinations.
Therefore,
\begin{align*}
&\int_{\mathbb{C}} |P_\alpha(r_k e^{j\theta_1}, y)| \exp\left(-\frac{|y|^2}{2\sigma^2}\right) dy \\
\leq&A^{|\alpha|}  \int_{\mathbb{C}} (1 + |y|^{2|\alpha|}) \exp\left(-\frac{|y|^2}{2\sigma^2}\right) dy < \infty.
\end{align*}
Based on the logarithmic terms in $\kappa(x_1, y)$, $|\kappa(x_1, y)|$ has at most polynomial growth in $|y|$.
Since the integral over $\theta_1 \in [0, 2\pi]$ is over a compact set, the product$
|\kappa(x_1, y)| \cdot (1 + |y|^{2|\alpha|}) \cdot \exp\left(-\frac{|y|^2}{2\sigma^2}\right)$
is integrable over $[0, 2\pi] \times \mathbb{C}$. Therefore, 
$M_\alpha \in L^1([0,2\pi] \times \mathbb{C})$, completing the proof. 
\end{IEEEproof}

By invoking Corollary~\ref{cor:derivative_domination} and the theorem on differentiation under the integral sign, it follows that $\omega_2(u, v)$ is smooth in $(u, v)$. Furthermore, the following proposition establishes its analyticity
\begin{proposition}
\label{prop:taylor_convergence}
For any point $(u_0, v_0) \in \mathbb{R}^2$ within the feasible BD input region, there exists a neighborhood of radius $\rho > 0$ such that the Taylor series
\begin{equation}
\omega_2(u_0 + \Delta_u, v_0 + \Delta_v) = \sum_{|\alpha|=0}^{\infty} \frac{1}{\alpha!}\frac{\partial^{|\alpha|}\omega_2}{\partial u^{\alpha_1}\partial v^{\alpha_2}}(u_0, v_0) \cdot \Delta_u^{\alpha_1}\Delta_v^{\alpha_2}
\end{equation}
converges absolutely for $\|(\Delta_u, \Delta_v)\| < \rho$. Thus, $\omega_2(u, v)$ is real-analytic on $\mathbb{R}^2$.
\end{proposition}

\begin{IEEEproof}
We establish the real analyticity of $\omega_2$ by proving the existence of a non-zero radius of convergence for its Taylor series. 
Fix $(u_0, v_0)$ and define the partial derivatives based on Corollary~\ref{cor:derivative_domination}. 
Using the Gaussian moment properties and the bounded support of $f_{X_1}$, we have
\begin{align*}
&\left|\frac{\partial^{|\alpha|}\omega_2}{\partial u^{\alpha_1}\partial v^{\alpha_2}}(u_0, v_0)\right| \leq \sum_{k=1}^{K}\iint M_\alpha(r_k\theta_1, y)\, dy\, d\theta_1 \\
\leq &\!\sum_{k=1}^{K}\!\frac{c_1 c_\alpha p_kA^{|\alpha|} }{2\pi^2\sigma^2}\!\! \!\iint\! |\kappa(r_k\theta_1, y)| (1 + |y|^{2|\alpha|}) \exp\!\left(\! \frac{-|y|^2}{2\sigma^2}\!\right)\! dy d\theta_1\\
\leq& \sum_{k=1}^{K}\frac{c_1 c_\alpha p_kA^{|\alpha|} }{2\pi^2 \sigma^2}\|\kappa\|_\infty \!\int_0^{2\pi} \!d\theta_1\!\int_{\mathbb C}\! (1 + |y|^{2|\alpha|}) \exp\!\left(\!\! \frac{-|y|^2}{2\sigma^2}\!\right)\! dy\\
\leq& \sum_{k=1}^{K}\frac{c_1 c_\alpha p_kA^{|\alpha|}}{\pi \sigma^2}   \|\kappa\|_\infty   (c_0+2\pi \sigma^{2|\alpha|+2} \cdot 2^{|\alpha|} \cdot |\alpha|!) \\
\leq& \sum_{k=1}^{K}\frac{c_1 c_\alpha p_kA^{|\alpha|}}{\pi \sigma^2} \|\kappa\|_\infty   (2|\alpha|)! \\
=& D \cdot (2|\alpha|)! \cdot A^{|\alpha|},
\end{align*}
where $D = \sum_{k=1}^{K}\frac{c_1 c_\alpha p_k \|\kappa\|_\infty}{\pi \sigma^2}$.
The Taylor series expansion around $(u_0, v_0)$ is bounded by
\begin{align*}
\mathcal{S} \triangleq \sum_{n=0}^{\infty} \sum_{\alpha_1+\alpha_2=n} \frac{1}{\alpha_1! \alpha_2!} \left| \partial^\alpha \omega_2 \right| |\Delta_u|^{\alpha_1} |\Delta_v|^{\alpha_2}.
\end{align*}
Applying the multinomial expansion and Stirling's approximation, we obtain
\begin{align*}
\mathcal{S}\leq& \sum_{|\alpha|=0}^{\infty} \frac{D \cdot (2|\alpha|)! \cdot A^{|\alpha|}}{\alpha_1! \alpha_2!} |\Delta u|^{\alpha_1}|\Delta v|^{\alpha_2} \nonumber\\
\leq& D \sum_{n=0}^{\infty} \frac{(2n)! \cdot A^n}{n!} \sum_{\alpha_1 + \alpha_2 = n} \frac{|\Delta u|^{\alpha_1}|\Delta v|^{\alpha_2}}{\alpha_1! \alpha_2! / n!} \nonumber\\
= & D \sum_{n=0}^{\infty} \frac{(2n)! \cdot A^n}{(n!)^2} (|\Delta u| + |\Delta v|)^n \nonumber\\
\overset{(a)}{\sim} & D \sum_{n=0}^{\infty} \frac{(4A)^n}{\sqrt{\pi n}} (|\Delta u| + |\Delta v|)^n,
\end{align*}
where (a) holds by using Stirling's approximation, i.e., $(2n)! \sim \sqrt{4\pi n}(2n/e)^{2n}$ and $(n!)^2 \sim 2\pi n (n/e)^{2n}$.
We can find that $\sum_{|\alpha|=0}^{\infty} \frac{\bar{c}_\alpha}{\alpha!} |\Delta u|^{\alpha_1}|\Delta v|^{\alpha_2}$ converges for $|\Delta u| + |\Delta v| < \frac{1}{4A}$. Therefore, the Taylor series converges in a neighborhood of $(u_0, v_0)$, establishing real analyticity.
\end{IEEEproof}

Based on the preceding analysis, $\omega_2(u, v)$ is a real-analytic function of the real and imaginary components of $X_2 = u + jv$ and non-constant.
Consequently, the support of the optimal BD input distribution, $\mathcal{S}_{f_{X_2}^*}$, is confined to the zero set of $\omega_2(u, v) - J_2$.
According to the properties of real-analytic sets, this zero set is a semi-analytic subset of $\mathbb{R}^2$. In general, such a set can only consist of
isolated points or pieces of real-analytic curves.
Critically, unless $\omega_2(u, v)$ is identically constant, its zero set has zero two-dimensional Lebesgue measure in $\mathbb{R}^2$. This formally proves that the optimal distribution $f_{X_2}^*$ cannot possess a continuous density over any area in the complex plane, but instead must be concentrated on discrete points or along specific geometric trajectories (e.g., concentric circles or arcs).

To further characterize the support, we observe that the contribution from $h(Y)$ depends on $X_2$ primarily through the distribution of the term $|a + X_2|$, which exhibits rotational symmetry along circles centered at $-a$. In contrast, the contribution from $h(Y|X_1)$ depends on $X_2$ via $|X_2|$, remaining invariant along circles centered at the origin. Therefore, the geometry of the optimal support $\mathcal{S}_{f_{X_2}^*}$ is governed by the interplay between these two distinct families of level sets: circles centered at $-a$ and those centered at the origin.

According to~\eqref{eq:opt_x2_2}, the support $\mathcal{S}_{f_{X_2}^*}$ of the optimal distribution must be contained within the set of points where $\Psi \triangleq\omega_2(X_2; f_{X_2}^*) - T_1 = 0$.
We decompose $\omega_2(X_2)$ into two components with distinct centers of symmetry:
\begin{align}\label{eq:symmetry}
\omega_2(X_2) = \Psi_1(|X_2 + a|^2) + \Psi_2(|X_2|^2),
\end{align}
where $\Psi_1 \triangleq -\int_{\mathbb{C}} \int_{\mathbb{C}} f_{X_1}(x_1) f_Z(y-(a+X_2)x_1) \mu_1\log f_Y(y) dy dx_1$ represents the contribution of $h(Y)$ and $\Psi_2 \triangleq -\int_{\mathbb{C}} \int_{\mathbb{C}} f_{X_1}(x_1) f_Z(y-(a+X_2)x_1)(\mu_2-\mu_1)\log f_{Y|X_1}(y|x_1) dy dx_1$ represents the contribution of $h(Y|X_1)$. Let $d_1 = |X_2 + a|^2 = (u+a)^2 + v^2$ and $d_2 = |X_2|^2 = u^2 + v^2$. If $\mathcal{S}_{f_{X_2}^*}$ contains a smooth curve $\mathcal{L}$, the gradient $\nabla \omega_2$ must vanish or remain perpendicular to the curve at every point on $\mathcal{L}$.

Calculating the gradient, we have
\begin{align}\label{eq:gradient}
\nabla \omega_2 = \Psi_1'(d_1) 
\begin{pmatrix} 2(u+a) \\ 2v \end{pmatrix} 
+ \Psi_2'(d_2) \begin{pmatrix} 2u \\ 2v 
\end{pmatrix}.
\end{align}
For the two vectors $\mathbf{v}_1 = [u+a, v]^T$ and $\mathbf{v}_2 = [u, v]^T$ to be linearly dependent, their cross product must be zero:
\begin{align}\label{eq:cross}
(u+a)v - uv = av = 0.
\end{align}

Given $a \neq 0$, this condition is satisfied only on the real axis, i.e., $v=0$. If $v \neq 0$, the gradients $\nabla d_1$ and $\nabla d_2$ are linearly independent since their cross product is not equal to zero from~\eqref{eq:cross}. A constant value of $\omega_2$ along a curve would require $\Psi_1' = \Psi_2' = 0$. However, since $\omega_2(X_2)$ is real-analytic and non-constant, if there were a continuous curve $\mathcal{L}$ such that $\Psi_1' = 0$ and $\Psi_2' = 0$ at every point on $\mathcal{L}$, the Identity Theorem would imply that $\Psi_1$ and $\Psi_2$ are constants over their entire domain. This directly contradicts the asymptotic divergence of the functional. Therefore, the set of points where the gradients are either linearly dependent or zero cannot form a curve, forcing the support $\mathcal{S}_{f_{X_2}^*}$ to be a set of isolated, and thus finite, discrete points.

If the support contains an interval on the real axis, i.e., $v=0$, then $\omega_2(u, 0) = \Psi_1((u+a)^2) + \Psi_2(u^2) = T_1$ on that interval. By the Identity Theorem for real-analytic functions of one variable, $\omega_2(u,0)$ would have to be constant for all $u \in \mathbb{R}$. This contradicts the established asymptotic divergence as $u \to \infty$.

Having excluded two-dimensional and one-dimensional structures, the zero set must consist solely of isolated points. 
Since the optimal distribution must reside within the compact set $\mathcal{D} = \{X_2: |X_2| \leq 1\}$, any set of isolated points within this bounded region must be finite by the Bolzano-Weierstrass Theorem.
Thus, the support $\mathcal{S}_{f_{X_2}^*}$ is a finite discrete set of points.

If $a=0$, the two centers of symmetry coincide at the origin, leading to $\nabla \omega_2 = [ \Psi_1'(|X_2|^2) + \Psi_2'(|X_2|^2) ] \mathbf{v}_2$. In this case, $\omega_2$ remains constant along any circle $|X_2| = \rho$, allowing for the concentric circle support structures mentioned previously. However, the presence of $a \neq 0$ breaks this rotational invariance, forcing the support to collapse into finite points.

\bibliographystyle{IEEEtran}
\bibliography{ref_AmBC}

\end{document}